\documentclass[letterpaper, 10 pt, conference]{ieeeconf}  

\IEEEoverridecommandlockouts                              
\usepackage{breqn}
\usepackage{mathtools}
\usepackage{tabularx}
\usepackage{xcolor}
\usepackage{amssymb}
\newcommand{\blue}[1]{\textcolor{blue}{#1}}
\newcommand{\red}[1]{\textcolor{red}{#1}}
\newcommand{\orange}[1]{\textcolor{orange}{#1}}
\newcommand{\magenta}[1]{\textcolor{magenta}{#1}}

\usepackage{hyperref}
\newcommand{\Supervisor}{\blue{\textit{Supervisor}}}
\newcommand{\Planner}{\red{\textit{Planner}}}
\newcommand{\Tracker}{\orange{\textit{Tracker}}}
\newcommand{\Customer}{\magenta{\textit{CustomerInterface}}}

\newtheorem{contract}{Contract}

\newtheorem{property}{Property}
\newtheorem{definition}{Definition}
\newtheorem{proposition}{Proposition}
\DeclarePairedDelimiter\abs{\lvert}{\rvert}%
\DeclarePairedDelimiter\norm{\lVert}{\rVert}%

\title{\LARGE \bf
Failure-Tolerant Contract-Based Design of an Automated Valet Parking System using a Directive-Response Architecture
}

\author{Josefine B. Graebener, Tung Phan-Minh, Jiaqi Yan, Qiming Zhao and Richard M. Murray
}

\begin{document}

\maketitle
\thispagestyle{empty}
\pagestyle{empty}

\begin{abstract}
Increased complexity in cyber-physical systems calls for modular system design methodologies that guarantee correct and reliable behavior, both in normal operations and in the presence of failures. This paper aims to extend the contract-based design approach using a directive-response architecture to enable reactivity to failure scenarios. The architecture is demonstrated on a modular automated valet parking (AVP) system. The contracts for the different components in the AVP system are explicitly defined, implemented, and validated against a Python implementation.
\end{abstract}

\section{INTRODUCTION}
Formally guaranteeing safe and reliable behavior for modern cyber-physical systems is becoming challenging as standard practices do not scale \cite{benveniste2018contracts}. Managing these highly complex architectures requires a design process that explicitly defines the dependencies and interconnections of system components to enable guaranteed safe behavior of the implemented system~\cite{censi2015mathematical}. A leading design methodology to develop component-based software is \textit{contract-based design}, which formalizes the design process in view of component hierarchy and composition \cite{filippidis2019decomposing,nuzzo2015platform,sangiovanni2012taming}. Contract-based design reduces the complexity of the design and verification process by decomposing the system tasks into smaller tasks for the components to satisfy. From the composition of these components, overall system properties can be inferred or proved. This contract-based architecture has been demonstrated for several applications \cite{damm2011using,damm2005boosting,nuzzo2013contract,maasoumy2015smart}. 
 Our goal here is to adapt and extend this framework to model a directive-response architecture on an automated valet parking system with the following features: 
\begin{enumerate}
    \item Discrete and continuous decision making components, which have to interact with one another.
    \item Different components have different temporal requirements.
    \item A natural hierarchy between the different 
    components in our system that may be thought of as different layers of abstraction.
    \item The system involves both human and non-human agents, the number of which is allowed to change over time.
    \item Industry interest in such a system.
\end{enumerate}
One example of industry efforts to commercialize such a system is the automated valet parking system developed by Bosch in collaboration with Mercedes-Benz, which has been demonstrated in the Mercedes-Benz Museum parking garage in Stuttgart, Germany. Bosch and Daimler also later announced in 2020 that they would set up a commercially operating AVP at the Stuttgart airport~\cite{BoschAVP}. Another commercial AVP system is supposed to be set up by Bosch in downtown Detroit as a collaboration with Bedrock and Ford~\cite{BoschAVPDetroit}.
Other examples include efforts by Siemens~\cite{siemens_avp} and DENSO~\cite{denso_avp}.
The contributions of this paper include the formulation of a formal contract structure for an automated valet parking system with multiple layers of abstraction with a directive-response architecture for failure-handling. By implementing this system in Python, we aim to bridge the large gap between abstract contract metatheory and such non-trivial engineering applications. In addition, we incorporate error handling into the contracts and demonstrate the use of this architecture and approach towards writing specifications in the context of the automated valet parking example. Finally, we prove that the composed implementation satisfies the composite contract, adding this example of a larger scale control system, involving a dynamic set of agents that are allowed to fail, to the small and slowly growing list of examples of formal assume-guarantee contract-based design.

\section{Theoretical Background}
\subsection{Contract Theory Background}
\label{contracts}
Contract-based design is a formal modular design methodology originally developed for component-based software systems \cite{bauer2012moving}. A component's behavior can be specified in terms of a guarantee that must be provided when its environment satisfies a certain assumption. This pairing of an assumption with a guarantee provides the basis for defining a contract. A contract algebra can be developed in which different contract operations can be defined which enable comparison between and combinations of contracts, formalizing modularity, reusability, hierarchy etc.~\cite{romeo2018quotient}. A comprehensive meta-theory of contracts is presented in~\cite{benveniste2018contracts}. In the following, we will introduce a variant of assume-guarantee contracts that incorporates a directive-response architecture.
\subsection{Directive-Response Architecture}
\label{dr_arch}
In a centralized approach for contingency management, recovery from failures is achieved by communicating with nearly every module in the system from a central module, hence increasing the system's complexity and potentially making it more error-prone \cite{wongpiromsarn2008distributed}.
The Mission Data system (MDS), developed by JPL as a multi-mission information and control architecture for robotic exploration spacecraft, was an approach to unify the space system software design architecture. MDS includes failure handling as an integral part of the design \cite{dvorak2000software,ingham2005engineering}. It is based on the state analysis framework, a system engineering methodology that relies on a state-based control architecture and explicit models of the system behavior. Fault detection in MDS is executed at the level of the modules, which report if they cannot reach the active goal and possible recovery strategies. Resolving failures is one of the tasks the system was designed to be capable of and not an unexpected situation \cite{dvorak2000software,rasmussen2001goal}.
Another architecture based on the state analysis framework is the Canonical Software Architecture (CSA) used on the autonomous vehicle Alice by the Caltech team in the DARPA Urban Challenge in 2007. The CSA enables decomposition of the planning system into a hierarchical framework, respecting the different levels of abstraction at which the modules are reasoning and the communication between the modules is via a directive-response framework \cite{burdick2007sensing}. This framework enables the system to detect and react to unexpected failure scenarios, which might arise from changes in the environment or hardware and software failures in the system \cite{wongpiromsarn2008distributed}.
In this paper we are trying to capture the MDS and CSA approaches by incorporating directive-response techniques into a contract framework.
  \begin{figure}
      \centering
      \includegraphics[width=3in]{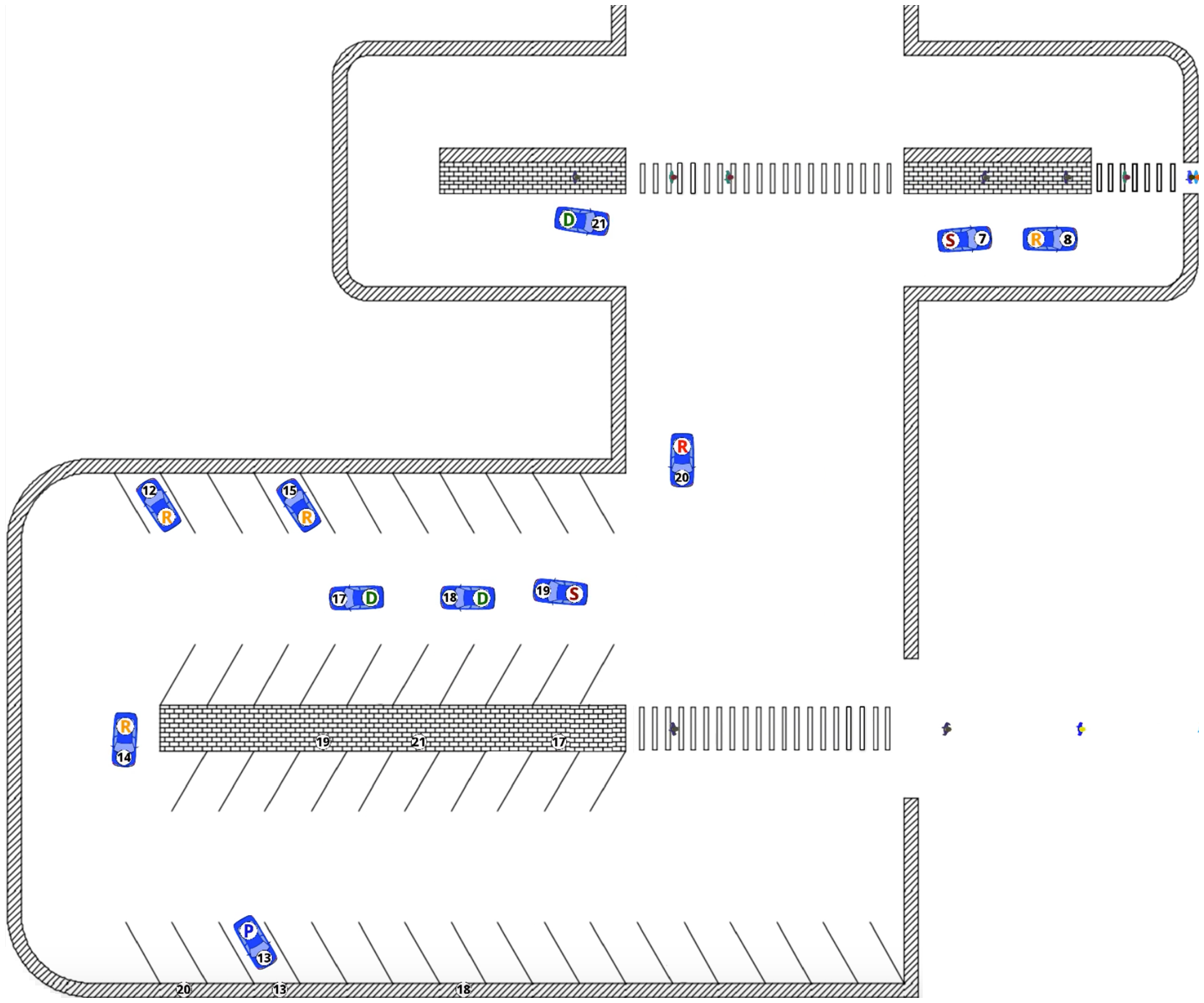}
      \caption{Snapshot from our AVP implementation showing human agents and vehicles as well as the parking lot topology.}
      \label{fig:lot}
  \end{figure}
\subsection{Directive-Response Contract Framework}
In this paper, we propose a contract-based design framework incorporating a directive-response architecture to enable reactivity to failures in the system.
System components can be abstracted as \textit{black boxes} constrained by assume-guarantee contracts that specify the behavior of the integrated system. 
Components communicate with one another by exchanging directives and responses, potentially acting according to a contingency plan that specifies how to react to possible failures. The higher module sends a directive, and the lower module chooses its responses according to its status in achieving the directive's intended goal. 
The system components are composed to satisfy the overall system requirements while interacting with the environment, such as safety and liveness specifications.
\section{Motivating Example}
The motivating example that we are developing in this paper is automated valet parking (AVP), as introduced in the previous section. The goal of this system is to automate the parking and retrieving process for multiple cars concurrently, while providing efficient operations in a safe environment.
\subsection{Overall Specification}
\label{AVPspec}
To be a successful operation, the AVP system needs to provide guarantees to customers regarding their safety and that their car will eventually be returned. These specifications can be written in linear temporal logic (LTL) \cite{pnueli1977temporal}. For a detailed discussion on LTL, see \cite{baier2008principles}. The $\square$ symbol represents the ``always'' operator and the $\lozenge$ represents ``eventually''. These are operators on predicates or traces.
An example of the specification is the following:
\begin{property}[Safety]
$\square \neg \texttt{collision}$
(Always no collision.)
\end{property}
and 
\begin{property}[Liveness]
$\square \texttt{healthy} \Rightarrow \lozenge \texttt{Returned}$
(Healthy car will eventually be returned.)
\end{property}
where the predicate $\texttt{collision}$ is \textit{True} if more than one car or pedestrian occupy the same space, and $\texttt{healthy}$ and $\texttt{Returned}$ are predicates which correspond to the status of the car, where $\texttt{healthy}$ is \textit{True} if the car does not have a failure and $\texttt{Returned}$ is \textit{True} once the control of the car has been given back to the customer. These specifications have to be satisfied for any implementation of the system and we will show this in our proof of the correctness of the composed system.
\section{Mathematical formulation}
To provide a formal description of the contracts and the components, we will introduce the mathematical background in this section. We will provide definitions regarding the geometry of the path planning, introduce the variables of our AVP world, and define the directive response framework and components.
\subsection{Geometry}
\begin{definition}[Path]
A \textit{path} is a continuous map $p: [0, 1] \rightarrow \mathbb{R}^2$. 
For each path $p$, let $p_h: [0, 1] \rightarrow (-180, 180]$ be such that $p_h(s)$ is the \textit{heading angle} measured in degrees from the abscissa to $p'(s)$, the derivative vector of $p$ with respect to $s$. For $t \in [0, 1]$, let $\tilde{p}(t)$ denote the element $p(t) \times p_h(t)$ of $\mathbb{R}^3$.
\end{definition}
We will denote the set of all paths by $\mathbf{P}$ and, by abuse of notation, we will also use $p$ to denote $p([0,1])$, the image of $[0,1]$ under $p$.
 \begin{definition}[Curvature feasibility]
Given $\kappa > 0$ and a path $p$, \textit{$\kappa$-feasible}$(p)$ is set to $\texttt{True}$ if and only if $p$ is twice differentiable on $[0,1]$, and its curvature $\frac{\abs{\det(p'(s), p''(s))}}{\norm{p'(s)}^3} < \kappa$ for $s \in [0, 1]$.
\end{definition}


\begin{definition}[$\delta$-corridor]
Let $\mathbb{B}\coloneqq\{\texttt{True}, \texttt{False}\}$.
If $p \in \mathbf{P}$, and $\delta: \mathbf{P} \times [0, 1] \times \mathbb{R}^3 \rightarrow \mathbb{B}$ is such that the corresponding subset: $$\Gamma_{\delta}(p) \coloneqq  \bigcup_{s \in [0,1]} \Gamma_{\delta}(p,s),$$ where $\Gamma_{\delta}(p,s) \coloneqq \{ (x,y,\theta) \in  \mathbb{R}^3 \mid \delta(p, s, (x,y,\theta)) = \texttt{True} \}$ such that $\Gamma_{\delta}(p, s)$ is open and contains $\tilde{p}(s)$ then we say $\Gamma_{\delta}(p)$ is a \textit{$\delta$-corridor} for $p$.
\end{definition}
\subsection{AVP World}

\textit{Building Blocks}:
In this section we will introduce naming symbols for objects that exist in the AVP world.
\begin{definition}[AVP World]
The \textit{AVP world} consists of the following
\begin{enumerate}
    \item A distinguished set of indexing symbols $\mathbf{T} := \{t, t', t'',...\}$ denoting time.
    \item A set of typed variables $\mathcal{U}$ to denote actions, states, channels, etc.
    \item The following set of constants: $\mathbf{C}$, $\mathbf{G}$ where
    \begin{enumerate}
        \item $\mathbf{C}$, a set of symbols, is called the customer set.
        \item $\mathbf{G}$, a set of symbols, is called the garage set containing the following constant values
        \begin{enumerate}
                \item $\mathbf{G}.\textit{drivable\_area}  \subseteq \mathbb{R}^3$, the set of configurations that vehicles are allowed to be in.
                \item $\mathbf{G}.\textit{walkable\_area}  \subseteq \mathbb{R}^2$, the area that pedestrians are allowed to walk on.
                \item         $\mathbf{G}.\textit{entry\_configurations} \subseteq \mathbb{R}^3$, a set of configurations that the customers can deposit their car in.
                \item $\mathbf{G}.\textit{return\_configurations}  \subseteq \mathbb{R}^3$, a set of configurations that the car should be returned in.
            \item $\mathbf{G}.\textit{parking\_spots} \in \mathbb{N}$, the number of parking spots available in the parking lot.
            \item $\mathbf{G}.\textit{interior} \subseteq \mathbb{R}^2$, the area inside the parking garage.
        \end{enumerate}
    \end{enumerate}
\end{enumerate}

\textit{Directive-Response Message Types:}
Each channel in the system is associated with a unique message type. The following are all the message types in our AVP system.\\
$\mathbf{A}(\cdot)$, directive types:
    \begin{enumerate}
        \item $\mathbf{A}(\Customer) \coloneqq \{\texttt{Park}, \texttt{Retrieve}\}$.
        \item $\mathbf{A}(\Supervisor) \coloneqq \mathbb{R}^6$.
        \item $\mathbf{A}(\Planner) \coloneqq \mathbf{P}$.
        \item $\mathbf{A}(\Tracker) \coloneqq \mathbf{I} \subseteq \mathbb{R}^2$, the set of all control inputs.
        \end{enumerate}
 $\mathbf{B}(\cdot)$, response types:
    \begin{enumerate}
        \item $\mathbf{B}(\Customer) \coloneqq \{ \texttt{Failed} \}$.
        \item $ \mathbf{B}(\Supervisor) \\\coloneqq \{\texttt{Rejected}, \texttt{Accepted}, \texttt{Returned}\}$.
        \item $ \mathbf{B}(\Planner) = \mathbf{B}(\Tracker) \\\coloneqq \{\texttt{Blocked}, \texttt{Failed}, \texttt{Completed}\}$.
    \end{enumerate}
For each type $\mathbf{T}$ we will denote by $\tilde{\mathbf{T}}$ the product type $\mathbf{T} \times \mathbf{C}$ which will be used to associate a message of type $\mathbf{T}$ with a specific customer in $\mathbf{C}$.
In addition, we will use $\mathbf{Id}$ to denote the set of message IDs.

\end{definition}

\textit{Behavior:}
For each variable $u \in \mathcal{U}$, we denote by $\text{type}(u)$ the \textit{type} of $u$, namely, the set of values that it can take. The types of elements of $\mathbf{T}$ are taken to be $\mathbb{R}_{\geq 0}$.
\begin{definition}[Behavior]
\label{behavior}
Let $Z$ be an ordered subset of variables in $\mathcal{U}$. A $Z$-\textit{behavior} is an element of $\mathcal{B}(Z) \coloneqq {(\prod_{z\in Z} \text{type}(z))}^{\mathbb{R}_{\geq 0}}$. Given $\sigma_Z \in \mathcal{B}(Z)$ and $\tau \in \mathbf{T}$, we will call $\sigma_Z(\tau)$ the \textit{valuation} of $Z$ at time $\tau$. If $z \in Z$, we will also denote by $z(\tau)$ the value of $z$ at time $\tau$.
\end{definition}
Note that each behavior in $Z \subseteq \mathcal{U}$ can be ``lifted'' to a set of behaviors in $\mathcal{U}$ by letting variables that are not contained in $Z$ assume all possible values in their domains. 
Additionally, the set of behaviors $\mathcal{B}(Z)$ can be lifted to a set of behaviors in $\mathcal{B}(\mathcal{U})$ in a similar way. To ease notational burden for the reader, we will take the liberty of not explicitly making any reference to the ``lifting'' operation in this paper when they are in use unless there is any ambiguity that may result from doing so.
\begin{definition}[Constraint]
A \textit{constraint} $k$ on a set of variables $Z$ is a function that maps each behavior of $Z$ to an element of $\mathbb{B}$, the Boolean domain. In other words, $k \in \mathbb{B}^{\mathcal{B}(Z)}$.
\end{definition}
Note that by ``lifting'', a constraint on a set of variables $Z$ is also a constraint on $\mathcal{U}$.
\begin{definition}[Channel variables]
For each component $X$ and another component $Y$, we can define two types of \textit{channel variables}:
\begin{itemize}
    \item $X_{\leftarrow Y}$, denoting an incoming information flow from $Y$ to $X$.
    \item $X_{\rightarrow Y}$, denoting an outgoing information flow from $X$ to $Y$.
\end{itemize}
In this work, we assume that $X_{\rightarrow Y}$ is always identical to $Y_{\leftarrow X}$.
Each channel variable must have a well-defined message type and each message $m$ has an ID denoted by $\textrm{id}(m) \in \textbf{Id}$. If the message has value $v$, then we will denote it by $[v, \textrm{id}(m)]$, but we will often refer to it as $[v]$ whereby we omit the ID part to simplify the presentation.
Intuitively, given a behavior, a channel variable $x$ is a function that maps each time step to the message the associated channel is broadcasting at that time step.
\end{definition}
\begin{definition}[System]
A \textit{system} $M$ consists of a set of each of the following
\begin{enumerate}
    \item internal variables/constants $\text{var}^M_{X}$,
    \item output channel variables $\text{var}^M_{Y}$,
    \item input channel variables $\text{var}^M_{U}$,
    \item constraints $\text{con}_M$ on $\text{var}^M_{X} \cup \text{var}^M_{Y} \cup \text{var}^M_{U}$.
\end{enumerate}
\end{definition}
A behavior of a system $M$ is an element of the set of behaviors that correspond to $\text{var}^M_{X} \cup \text{var}^M_{Y} \cup \text{var}^M_{U}$ subject to $\text{con}_M$. This is denoted by $\mathcal{B}(M)$.

\textit{Directive-response:}
Before introducing directive-response systems, for any predicates $A$ and $B$, we define the following syntax:
\begin{equation}
A \leadsto B \coloneqq \forall t:: A(t)  \Rightarrow \exists t' \geq t :: B(t').
\tag{``leads to''}
\end{equation}
\begin{equation}
        A \preceq B \coloneqq  \forall t:: B(t)   \Rightarrow \exists t' \leq t :: A(t'). 
        \tag{``precedes''}
\end{equation}
\begin{equation}
\square_{\geq t} A \coloneqq \forall t' \geq t :: A(t').
\tag{``always from $t$''}
\end{equation}
\begin{equation}
\begin{split}
\mathrm{starts\_at}(A, t) \coloneqq A(t) \land
 \forall t' < t :: \lnot A(t').
\end{split}
\end{equation}

If $M$ is a set-valued variable, then we define:
\begin{equation}
\begin{split}
\mathrm{persistent}(M) \coloneqq \forall t:: \forall m :: m \in M(t)
\Rightarrow \square_{\geq t}(m \in M).
\end{split}
\end{equation}
\begin{definition}[Directive-response system]
A~\textit{directive-res-ponse system} $M$ is a system such that for each output (resp., input) channel variable $chan$ there is an internal variable $\textit{send}_{chan}$ (resp., $\textit{receive}_{chan}$) whose domain is a collection of sets of messages that are of the type associated with $chan$.
If $chan$ is an output channel variable, there is a causality constraint $k_{chan} \in \text{con}_M$ defined by:
    \begin{equation}
    \begin{split}
    k_{chan} \coloneqq m \in \textit{send}_c \preceq m = chan. 
    \end{split}
    \end{equation}
That is, a message must be sent before it shows in the channel. Otherwise if $chan$ is an input channel variable:
    \begin{equation}
    \begin{split}
    k_{chan} \coloneqq m = chan \preceq m \in \textit{receive}_{chan}.
    \end{split}
    \end{equation}
\end{definition}
Namely, a message cannot be received before it is broadcasted.
\begin{definition}[Lossless directive-response system]
A lossless directive-response system is a directive-response system such that if $chan$ is an output channel then
\begin{equation}
\begin{split}
\textrm{persistent}(\text{send}_{chan}) \land (m \in \text{send}_{chan} \leadsto m = chan).
\end{split}
\end{equation}
and if $chan$ is an input channel
\begin{equation}
\begin{split}
\textrm{persistent}(\text{receive}_{chan}) \land (m = chan \leadsto m \in \text{receive}_{chan}).
\end{split}
\end{equation}
\end{definition}
\begin{definition}[Assume-guarantee contracts]
An \textit{assume-guarantee contract} $\mathcal{C}$ for a directive-response system $M$ consists of a pair of behaviors $A$, $G$ of $M$ and denoted by $\mathcal{C} = (A,G)$. An \textit{environment} for $\mathcal{C}$ is any set of all behaviors that are contained in $A$ while an \textit{implementation} of $\mathcal{C}$ is any set of behaviors that is contained in $A \Rightarrow G$. $\mathcal{C}$ is said to be \textit{saturated} if the guarantee part satisfies $G = (\lnot A \lor G) = (A \Rightarrow G)$.
\end{definition}
Note that any contract can be converted to the saturated form without changing its sets of environments and implementations. The saturated form is useful in making contract algebra less cumbersome in general. If $M$ is a system, then we say $M$ satisfies $\mathcal{C}$ if $\mathcal{B}(M) \subseteq (A \Rightarrow G)$. Furthermore, the system composition $M_1 \times M_2$ of $M_1$ and $M_2$ is a system whose behavior is equal to $\mathcal{B}(M_1) \cap \mathcal{B}(M_2)$.
\begin{definition}[Customer]
A \textit{customer} is an element of $\mathbf{C}$. Corresponding to each $c \in \mathbf{C}$ is a set of $\mathcal{U}$ variables $\text{var}(c)$ that include $c.x$, $c.y$ (the coordinates of the customer him/herself), $\textit{c.car.x}$, $\textit{c.car.y}$, $\textit{c.car.}\theta$ (the coordinates and heading of the customer's car), $\textit{c.car.healthy}$, whether the car is healthy, $\textit{c.controls.v}$, $\textit{c.controls.}\varphi$ (the velocity and steering inputs to the vehicle), $\textit{c.car.}\ell$ (the length of the car), $\textit{c.car.towed}$ (whether the car is being towed). We will use the shorthand $\textit{c.car.state}$ to mean the 3-tuple $(\textit{c.car.x, c.car.y, c.car}.\theta)$.
\end{definition}
For each behavior in $\mathcal{B}(\mathcal{U})$, we require each $c \in \mathbf{C}$ for which $\textit{c.car.towed}$ is $\textit{False}$ to satisfy the following constraints that describe the Dubins car model:
\begin{align}
\label{dubins}
\begin{split}
    \frac{\textit{d(c.car.x})}{\textit{dt}}(t) &= \textit{c.controls.v}(t) \cos(\textit{c.car.}\theta(t)) \\
    \frac{\textit{d(c.car.y})}{\textit{dt}}(t) &= \textit{c.controls.v}(t) \sin(c.car.\theta(t)) \\
    \frac{\textit{d(c.car.}\theta)}{\textit{dt}}(t) &= \frac{\textit{c.controls.v}(t)}{\textit{c.car.}\ell} \tan(\textit{c.controls}.\varphi(t))
\end{split}
\end{align}

\subsection{AVP System}
By treating the $\Customer$ as an external component, the AVP system consists of three internal components: \Supervisor, \Planner\: and \Tracker. These systems are described below. 
\subsubsection{\Customer}
The environment in which the system shall operate consists of the customers and the pedestrians, which we will call a $\Customer$. A customer drops off the car at the drop-off location and is assumed to make a request for the parked car back from the garage eventually. The pedestrians are also controlled by the environment. When a pedestrian was generated by the environment, they start walking on the crosswalks. Pedestrians are confined to the pedestrian path, meaning they will not leave the crosswalk and walkway areas and their dynamics are continuous, meaning no sudden jumps. The cars move according to their specified dynamics. This includes a breaking distance depending on their velocity and maximum allowed curvature. For a formal description, refer to Table~\ref{customer_table}. Below are some constraints we impose on this module.
\vspace{-2mm}
\begin{table}
    \caption{$\Customer$ directive-response system.}
    \label{customer_table}
    \scriptsize
    \begin{tabularx}{\columnwidth}{|XX|}
        \hline

        \textbf{Internal variables/constants} $\text{var}_{X}$& \\
        \hline
        $\mathbf{C}$ & The set of all customers in the AVP world. \\
        \hline
        
        \textbf{Outputs} $\text{var}_{Y}$ & \\
        \hline
        $\Customer_{\rightarrow \Supervisor}$ & An output channel of type $\tilde{\mathbf{A}}(\Customer)$. \\
        \hline
        
        \textbf{Inputs} $\text{var}_{U}$& \\
        \hline
        $\Customer_{\leftarrow \Supervisor}$ & An input channel of type $\tilde{\mathbf{B}}(\Supervisor)$. \\
        $\Customer_{\leftarrow \Tracker}$ & An input channel of type $\tilde{\mathbf{A}}(\Tracker)$. \\
        \hline
        
        \textbf{Constraints} $\text{con}_M$ & \\
        \hline
        Vehicle dynamics & See \eqref{dubins} \\
        Car and pedestrian limits & \eqref{control_limits} and \eqref{pedestrian_limits}. \\
        \hline
    \end{tabularx}
\end{table}

\begin{equation}
\label{control_limits}
    \begin{split}
        \forall c \in \mathbf{C} :: \square (v_{\min} \leq \textit{c.controls.v} \land \textit{c.controls.v} \leq v_{\max} \\
        \land \varphi_{\min} \leq \textit{c.controls}.\varphi \land \textit{c.controls}.\varphi \leq \varphi_{\max})
    \end{split}
\end{equation}
\begin{equation}
\label{pedestrian_limits}
    \begin{split}
        \forall c \in \mathbf{C} :: \forall s. \Big\lVert \Big(\frac{d(c.x)}{dt}(s), \frac{d(c.y)}{dt}(s) \Big) \Big\rVert \leq v_{ped, \max}.
    \end{split}
\end{equation}
\subsubsection{\Supervisor} 
A $\Supervisor$ component is responsible for the high level decision making in the process. It receives the $\Customer\:$ requests and processes them by sending the appropriate directives to the $\Planner$ to fulfill a task. A $\Supervisor$ determines whether a car can be accepted into the garage or rejected. It also receives responses from the $\Planner$. A $\Supervisor$ is to be aware of the reachability, the vacancy, and occupied spaces in the lot, as well as the parking lot layout. Formally, a $\Supervisor$ is a lossless directive-response system described by Table~\ref{supervisor_table}.
\begin{table}
    \caption{$\Supervisor$ directive-response system.}
    \label{supervisor_table}
    \scriptsize
    \begin{tabularx}{\columnwidth}{|XX|}
        \hline

        \textbf{Internal variables/constants} $\text{var}_{X}$ & \\
        \hline
        $\mathbf{G}.*$ & All $\mathbf{G}$ objects.\\
        $num\_active\_customers$ & The number of cars currently being served in the parking lot. \\
        \hline
        
        \textbf{Outputs} $\text{var}_{Y}$ & \\
        \hline
        $\Supervisor_{\rightarrow \Customer}$ & An output channel of type $\tilde{\mathbf{B}}(\Supervisor)$. \\
        $\Supervisor_{\rightarrow \Planner}$ & An output channel of type $\tilde{\mathbf{A}}(\Supervisor)$. \\
        \hline
        
        \textbf{Inputs} $\text{var}_{U}$ & \\
        \hline
        $\Supervisor_{\leftarrow \Customer}$ & An input channel of type $\tilde{\mathbf{A}}(\Customer)$. \\
        $\Supervisor_{\leftarrow \Planner}$ & An input channel of type $\tilde{\mathbf{B}}(\Planner)$. \\
        \hline
        
        \textbf{Constraints} $\text{con}_M$& \\
        \hline
        Parking lot topology & Any specific geometric constraints on $\mathbf{G}.*$.\\
        Number of active customers & $num\_active\_customers$ must be equal to the number of cars that have been accepted but not yet left the parking lot.\\
        \hline
    \end{tabularx}
\end{table}

\subsubsection{\Planner} A $\Planner$ system receives directives from the \Supervisor\: to make a car reach a specific location in the parking lot. A $\Planner$ system has access to a planning graph determined from the parking lot layout, and thus can generate executable trajectories for the cars to follow. The $\Planner$ is aware of the locations of the agents and the obstacles in the parking lot from the camera system. A $\Planner$ is a lossless directive-response system described by Table~\ref{planner_table}.

\begin{table}
    \scriptsize
    \caption{$\Planner$ directive-response system.}
    \label{planner_table}
    \begin{tabularx}{\columnwidth}{|XX|}
        \hline
        \textbf{Interval variables/constants} $\text{var}_{X}$ & \\
        \hline
        $\mathbf{G}.*$ & All $\mathbf{G}$ objects. \\
        $\{c.car.x, c.car.y, c.car.\theta \mid c \in \mathbf{C} \}$ &  The configurations of all cars in AVP world. \\
        $\kappa$ & Maximum allowable curvature. \\
        \hline
        
        \textbf{Outputs}  $\text{var}_{Y}$&\\
        \hline
        $\Planner_{\rightarrow \Supervisor}$ & An output channel of type $\tilde{\mathbf{B}}(\Planner)$. \\ 
        $\Planner_{\rightarrow \Tracker}$ & An output channel of type $\tilde{\mathbf{A}}(\Planner)$. \\
        \hline
        
        \textbf{Inputs} $\text{var}_{U}$&  \\
        \hline
        $\Planner_{\leftarrow \Supervisor}$ & An input channel of type $\tilde{\mathbf{A}}(\Supervisor)$. \\
        $\Planner_{\leftarrow \Tracker}$ & An input channel of type $\tilde{\mathbf{B}}(\Tracker)$. \\ 
        \hline
        
        \textbf{Constraints} $\text{con}_{M}$& \\
        \hline
        Parking lot topology & Any specific geometric constraints on $\mathbf{G}.*$.\\
        $\kappa$ & Maximum allowable curvature given car dynamics and input constraints. \\
        \hline
    \end{tabularx}
\end{table}

\subsubsection{\Tracker} A $\Tracker$ system is responsible for the safe control of cars that are accepted into the garage by a $\Supervisor$. It receives directives from a $\Planner$ consisting of executable paths to track and send responses based on the task status to a $\Planner$. See Table~\ref{tracker_table}.

\begin{table}
    \caption{$\Tracker$ directive-response system.}
    \label{tracker_table}
    \scriptsize
    \begin{tabularx}{\columnwidth}{|XX|}
        \hline
                
        \textbf{Interval variables/constants}  $\text{var}_{X}$& \\
        \hline
         $\delta$ & Corridor map.\\
         $\varepsilon_{\min,car}$ & Minimum safety distance to other cars.\\
         $\varepsilon_{\min, people}$ & Minimum safety distance to pedestrians.\\
        \hline
        
        \textbf{Outputs}  $\text{var}_{Y}$ &  \\
        \hline
        $\Tracker_{\rightarrow \Planner}$ & An output channel of type $\tilde{\mathbf{B}}(\Tracker)$. \\         $\Tracker_{\rightarrow \Customer}$ & An output channel of type $\tilde{\mathbf{A}}(\Tracker)$. \\ 

        \hline
        
        \textbf{Inputs}  $\text{var}_{U}$&  \\
        \hline
        $\Tracker_{\leftarrow \Planner}$ & An input channel  of type $\tilde{\mathbf{A}}(\Planner)$. \\ 
        \hline

        \textbf{Constraints}  $\text{con}_{M}$& \\
        \hline
        Corridor constraints & 
        In our implementation, we define the $\delta$-corridor for any path $p$ to be the open set containing points whose distance to the closest point in $p$ does not exceed 3 meters. \\
        $\varepsilon_{\min,car}$, $\varepsilon_{\min,people}$ & These values are determined based on the dynamics and the uncertainty $\Delta_{Car}$.\\
        \hline
    \end{tabularx}
\end{table}

\begin{figure}
      \centering
      \includegraphics[width=\linewidth*3/4]{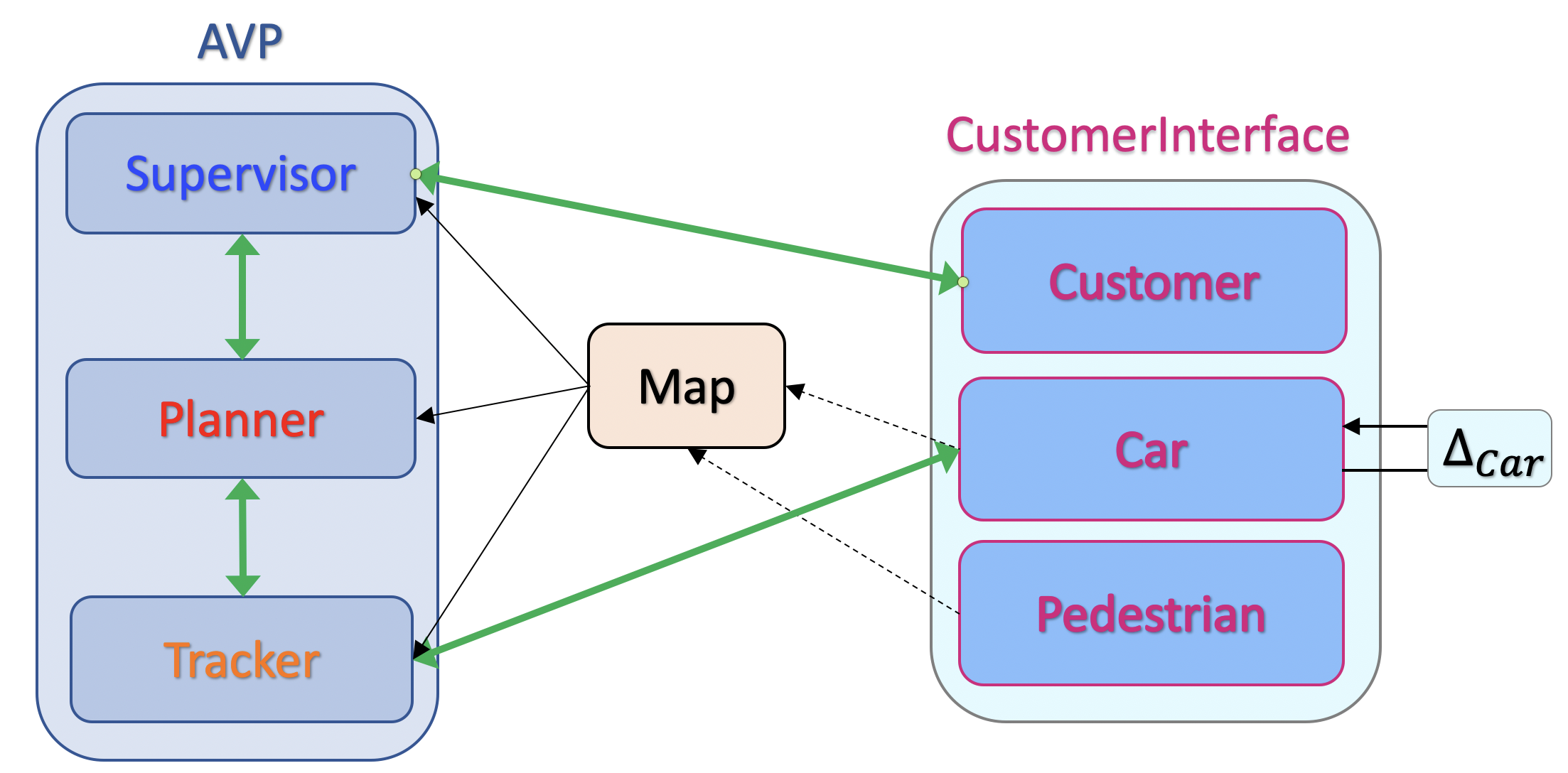}
      \caption{This figure depicts the contracts and components in the AVP system. The green arrows represent directive-response assume-guarantee contracts, solid black arrows represent communication, and dashed black arrows represent passive information flow (observing movement of the agents). The $\Delta$ in the car component represents the possibility of failure and uncertainty.}
      \label{fig:contracts}
\end{figure}
\section{AVP Contracts}
In this section we will define the contracts for each of the modules in our system. These contracts are the guidelines for the implementation, and will be used to verify each of the components, as well as the composed system.
In Figure~\ref{fig:contracts} the green arrows represent directive-response assume-guarantee contracts, solid black arrows represent communication, and dashed black arrows represent passive information flow (observing movement of the agents). The $\Delta$ in the car component represents the possibility of failure and uncertainty.
\begin{contract}[$\mathcal{C}_{\Customer}$]
The following is the contract for the \Customer.
\begin{itemize}
    \small
    \item Assumes
    \begin{itemize}
        \item If the \Customer\:sends a request to the $\Supervisor$, then they will receive a response from the $\Supervisor$:
\begin{equation}
\label{customer:assume:request_always_receive_response}
\begin{split}
&\forall c \in \mathbf{C}::([m, c] \in \text{send}_{\Customer_{\rightarrow \Supervisor}} \\&\leadsto \exists r \in \mathbf{B}(\Supervisor) ::\\ &[r, c] \in 
\text{receive}_{\Customer_{\leftarrow \Supervisor}}).
\end{split}
\end{equation}
        \item If the car is healthy and accepted by the garage, it will be returned after being summoned:
\begin{equation}
\label{customer:assume:accepted_and_healthy_and_requested_then_returned}
\begin{split}
   &\forall c \in \mathbf{C} :: (\square_{\geq0}(\textit{c.car.healthy}) \\ &\land ([\texttt{Accepted}, c] \in \text{receive}_{\Customer_{\leftarrow \Supervisor}} \\&\land
   [\texttt{Retrieve}, c] \in \text{send}_{\Customer_{\rightarrow \Supervisor}}) \leadsto \\ 
   &[\texttt{Returned}, c] \in \text{receive}_{\Customer_{\leftarrow \Supervisor}}
   ).
\end{split}
\end{equation}
    \end{itemize}
    \item Guarantees
    \begin{itemize}
        \item When the request is accepted, the \Customer\:should not tamper with the car controls until the car is returned (i.e., control signals should match the directive) :
\begin{equation}
\label{customer:guarantee:no_tampering}
    \begin{split}
    &\forall c \in \mathbf{C}:: \forall t:: \forall (v, \varphi) \in \mathbf{I} :: \\ 
    &([\texttt{Accepted}, c] \in \text{receive}_{\Customer_{\leftarrow \Supervisor}}(t) \\ 
    &\land \lnot ([\texttt{Returned}, c] \in \text{receive}_{\Customer_{\leftarrow \Supervisor}}(t))  \\
    &\Rightarrow([(v, \varphi), c] \in \text{receive}_{\Customer_{\leftarrow \Tracker}}(t) \\ &\land \forall t' < t :: [(v, \varphi), c] \not \in  \text{receive}_{\Customer_{\leftarrow \Tracker}}(t') \\
    &\Rightarrow \textit{c.controls.v}(t) = v \land \textit{c.controls}.\varphi(t) = \varphi)).
    \end{split}
\end{equation}
    \item When the \Customer\: is not receiving any new input signal, then it keeps the control inputs at zero:
\begin{equation}
    \label{customer:guarantee:no_tracker_input}
    \begin{split}
    &\forall c \in \mathbf{C}:: \forall t:: \\ 
    &([\texttt{Accepted}, c] \in \text{receive}_{\Customer_{\leftarrow \Supervisor}}(t)\land \\ 
    & \lnot ([\texttt{Returned}, c] \in \text{receive}_{\Customer_{\leftarrow \Supervisor}}(t)) \Rightarrow \\
    &(\forall (v, \varphi) \in \mathbf{I} ::  [(v, \varphi), c] \in \text{receive}_{\Customer_{\leftarrow \Tracker}}(t) \\ 
    &\Rightarrow \exists t' < t :: [(v, \varphi), c] \in  \text{receive}_{\Customer_{\leftarrow \Tracker}}(t')) \\
    &\Rightarrow\textit{c.controls.v}(t) = 0 \land \textit{c.controls}.\varphi(t) = 0))).
    \end{split}
\end{equation}
        \item From sending a request until receiving a response, the car must stay in the deposit area:
    \begin{equation}
        \label{customer:guarantee:stay_in_box}
        \begin{split}
            &\forall c \in \mathbf{C} :: \square_{\geq0}([\texttt{Park}, c] \in \text{send}_{\Customer_{\rightarrow \Supervisor}} \\
            &\land  [\texttt{Accepted}, c] \not\in \text{receive}_{\Customer_{\leftarrow \Supervisor}} \\ 
            &\land [\texttt{Rejected}, c] \not\in \text{receive}_{\Customer_{\leftarrow \Supervisor}} \\
            &\Rightarrow \textit{c.car.state} \in \mathbf{G}.\textit{entry\_configurations}).
        \end{split}
    \end{equation}
        \item After the car is deposited, the customer will eventually summon it:
    \begin{equation}
    \label{customer:guarantee:will_pick_up}
        \begin{split}
            &\forall c \in \mathbf{C}:: [\texttt{Accepted}, c] \\ 
            &\in \text{receive}_{\Customer_{\leftarrow \Supervisor}} \leadsto \\
            &[\texttt{Retrieve}, c] \in \text{send}_{\Customer_{\leftarrow \Supervisor}}.
        \end{split}
    \end{equation}
        \item Pedestrians will only walk on ``walkable'' area:
        \begin{equation}
        \label{customer:guarantee:on_walkable_area}
        \begin{split}
            \forall c \in \mathbf{C}:: \square_{\geq0} ((c.x, c.y) \in \mathbf{G}.\textit{walkable\_area}).
        \end{split}
        \end{equation}
        \item Pedestrians will not stay on crosswalks forever:
        \begin{equation}
        \label{customer:guarantee:transient_crosswalks}
        \begin{split}
            &\forall c \in \mathbf{C}:: ((c.x, c.y) \\
            &\in \mathbf{G}.\textit{walkable\_area} \cap \mathbf{G}.\textit{drivable\_area} \\
            &\leadsto (c.x, c.y) \not\in \mathbf{G}.\textit{walkable\_area} \cap \mathbf{G}.\textit{drivable\_area}).
        \end{split}
        \end{equation}
        \item If the car is not healthy and not towed, it cannot move:
        \begin{equation}
        \label{customer:guarantee:unhealthy_cannot_move}
            \begin{split}
        &\forall c \in \mathbf{C} :: \square_{\geq0}(\lnot \textit{c.car.healthy} \land \lnot \textit{c.car.towed} \Rightarrow  \\ 
        &\textit{c.controls.v} = 0 \land \textit{c.controls}.\varphi = 0).
            \end{split}
        \end{equation}
        \item Sending a \texttt{Retrieve} message must always be preceded by receiving an \texttt{Accepted} message from the $\Supervisor$:
        \begin{equation}
        \label{customer:guarantee:must_receive_accepted_before_retrieve}
            \begin{split}
                &\forall c \in \mathbf{C} :: [\texttt{Accepted}, c] \in \text{receive}_{\Customer_{\leftarrow \Supervisor}} \\ 
                &\preceq [\text{Retrieve}, c] \in \text{send}_{\Customer_{\rightarrow \Supervisor}}.
            \end{split}
        \end{equation}
        \item If a customer receives \texttt{Rejected} or \texttt{Returned} from the $\Supervisor$, then they must leave the lot forever:
        \begin{equation}
        \label{customer:guarantee:rejected_leave_forever}
            \begin{split}
                &\forall c \in \mathbf{C} :: \forall t :: [\texttt{Rejected},c] \\ 
                &\in \text{receive}_{\Customer_{\leftarrow \Supervisor}}(t) \\ 
                &\lor [\texttt{Returned},c] \in \text{receive}_{\Customer_{\leftarrow \Supervisor}}(t) \Rightarrow \\ 
                &\exists t' > t:: \square_{\geq t'}((\textit{c.car.x}, \textit{c.car.y}) \not \in \mathbf{G}.\textit{interior}).
            \end{split}
        \end{equation}
    \end{itemize}
\end{itemize}
\end{contract}

\begin{contract}[$\mathcal{C}_{\Supervisor}$] The contract for the $\Supervisor$ is as follows.
\begin{itemize}
    \small
    \item Assumes
    \begin{itemize}
        \item Towing eventually happens after the $\Supervisor$ is alerted of car failure:
        \begin{equation}
        \label{supervisor:assume:all_broken_will_be_towed}
        \begin{split}
            &\forall c \in \mathbf{C}:: \forall t:: [\texttt{Failed},c] \in \text{receive}_{\Supervisor_{\leftarrow \Planner}}(t) \Rightarrow \\
            &\exists t':: \square_{\geq t'}(
            \textit{c.car.towed} \land (\textit{c.car.x}, \textit{c.car.y}) \\ &\not\in \mathbf{G}.\textit{interior}).
        \end{split}
        \end{equation}
        \item If a car fails, then the $\Planner$ reports \texttt{Failed}:
        \begin{equation}
        \label{supervisor:assume:fail_report}
    \begin{split}
        &\forall c \in \mathbf{C} :: \lnot \textit{c.car.healthy} \\ &\leadsto
        [\texttt{Failed}, c] \in \text{receive}_{\Supervisor_{\leftarrow \Planner}}.
    \end{split}
\end{equation}
        \item Cars making requests are deposited correctly by the customer:
        \begin{equation}
        \label{supervisor:assume:cowlagi}
            \begin{split}
                &\forall c \in \mathbf{C} :: \square_{\geq0} ([\texttt{Park}, c] \in \text{receive}_{\Supervisor_{\leftarrow \Customer}} \\
                &\land ([\texttt{Accepted}, c] \not\in \text{send}_{\Supervisor_{\rightarrow \Customer}} \\
                &\lor [\texttt{Rejected}, c] \not\in \text{send}_{\Supervisor_{\rightarrow \Customer}})\\
                &\Rightarrow
                \textit{c.car.state} \in \mathbf{G}.\textit{entry\_configurations}).
            \end{split}
        \end{equation}
    \item If a car is healthy and summoned, then it will eventually appear at the return area and the $\Planner$ will send a \texttt{Completed} signal to the $\Supervisor$:
    \begin{equation}
    \label{supervisor:assume:planner_reports_complete}
        \begin{split}
   &\forall c \in \mathbf{C} :: (\square_{\geq0} \textit{c.car.healthy} \land
   [\texttt{Retrieve}, c] \\ &\in \text{receive}_{\Supervisor_{\leftarrow \Planner}} 
   \\ &\leadsto ([\texttt{Completed}, c] \in \text{receive}_{\Supervisor_{\leftarrow \Planner}}  \\ &\land\textit{c.car.state} \in \mathbf{G}.\textit{return\_configurations})).        \end{split}
    \end{equation}
    
    \end{itemize}
    \item Guarantees
    \begin{itemize}
        \item All requests from customers will be replied:
        \begin{equation}
        \label{supervisor:guarantee:responsive}
            \begin{split}
                &\forall c \in \mathbf{C} :: ([m, c] \in \text{receive}_{\Supervisor_{\leftarrow \Customer}} \leadsto \\
                & \exists r \in \textbf{B}(\Supervisor):: [r,c] \in \text{send}_{\Supervisor_{\rightarrow \Customer}}).
            \end{split}
        \end{equation}
        \item The $\Supervisor$ cannot send a \texttt{Returned} message to the \Customer\: unless it has received a \texttt{Completed} message from the $\Planner$ and the car is in the return area:
        \begin{equation}
        \label{supervisor:guarantee:only_send_returned_when_really_returned}
            \begin{split}
                & \forall c \in \mathbf{C} :: 
                [\texttt{Completed}, c] \in  \text{receive}_{\Supervisor_{\leftarrow \Planner}} \\ & \land
                \textit{c.car.state} \in
                \mathbf{G}.\textit{return\_configurations} \preceq \\
                 & [\texttt{Returned}, c] \in \text{send}_{\Supervisor_{\rightarrow \Customer}}.
            \end{split}
        \end{equation}
        
        \item If a car is healthy and a \texttt{Retrieve} message is received, then the last thing sent to the $\Planner$ should be a directive to the return area (the second configuration should be one of the return configurations).
        \begin{equation}
    \label{supervisor:guarantee:healthy_retrieval_implies_last_command_is_to_return_area}
        \begin{split}
   & \forall c \in \mathbf{C} :: (\square_{\geq0} \textit{c.car.healthy} \Rightarrow
   [\texttt{Retrieve}, c] \\ & \in \text{receive}_{\Supervisor_{\leftarrow \Planner}} \land \\
   & \exists p_0, p_1 \in \mathbb{R}^3 :: [(p_0, p_1), c] \in \text{send}_{\Supervisor_{\rightarrow \Planner}} \\ &\land
   \forall p_0', p_1' \in \mathbb{R}^3 :: [(p_0', p_1'), c] \in \text{send}_{\Supervisor_{\rightarrow \Planner}} \\ & \preceq [(p_0, p_1), c] \in \text{send}_{\Supervisor_{\rightarrow \Planner}} \\
   & \Rightarrow p_1 \in \mathbf{G}.\textit{return\_configurations}.        \end{split}
    \end{equation}
        
        \item If the car is healthy and if it is ever summoned, then the $\Supervisor$ will send a \texttt{Returned} message to its owner:
        \begin{equation}
        \label{supervisor:guarantee:eventual_return}
            \begin{split}
                &\forall c \in \mathbf{C} :: (\square_{\geq0} \textit{c.car.healthy} \Rightarrow \\
                &[\texttt{Retrieve}, c] \in \text{receive}_{\Supervisor_{\leftarrow \Customer}} \\       &\leadsto[\texttt{Returned},c] \in \text{send}_{\Supervisor_{ \rightarrow \Customer}}).
            \end{split}
        \end{equation}
        \item If there is a not-yet-responded-to \texttt{Park} request and the parking lot capacity is not yet reached, then the $\Supervisor$ should accept the request:
        \begin{equation}
        \label{supervisor:guarantee:acceptance}
            \begin{split}
              & \forall c \in \mathbf{C} :: \forall t::  \exists [\texttt{Park}, c] \\ & \in \text{receive}_{\Supervisor_{\leftarrow \Customer}}(t) \\
              & \land  \forall t' \leq t:: [\texttt{Rejected}, c] \not\in \text{send}_{\Supervisor_{\rightarrow \Customer}}(t') \\
              & \land [\texttt{Accepted}, c] \not\in \text{send}_{\Supervisor_{\rightarrow \Customer}}(t') \\
              & \land \textit{num\_active\_customers}(t) < \mathbf{G}.\textit{parking\_spots} \\ & \Rightarrow \exists t'' > t :: [\texttt{Accepted}, c] \\ & \in \text{send}_{\Supervisor_{\rightarrow \Customer}}(t'').
            \end{split}
        \end{equation}
        \item For every \texttt{Accepted} to or \texttt{Retrieve} from the $\Customer$ or \texttt{Blocked} from the $\Planner$, the $\Supervisor$ sends a pair of configurations to the $\Planner$, the first of which is the current configuration of the car and such that there exists a path of allowable curvature :
        \begin{equation}
        \label{supervisor:guarantee:send_pair_of_configurations}
            \begin{split}
                & \forall c \in \mathbf{C}:: [\texttt{Accepted}, c] \in \text{send}_{\Supervisor_{\rightarrow \Customer}} \\ & \lor [\texttt{Retrieve}, c] \in \text{receive}_{\Supervisor_{\leftarrow \Customer}} \\
                & \lor [\texttt{Blocked}, c] \in \text{receive}_{\Supervisor_{\leftarrow \Planner}}
                \leadsto \\ & \exists k_0, k_1 \in \mathbb{R}^3 :: [(k_0, k_1), c] = \Supervisor_{\rightarrow \Planner} \\ & \land 
                k_0 = \textit{c.car.state} \land \exists p \in \mathbf{P} :: \text{$\kappa$-feasible}(p) \land \\
                & \tilde{p}(0) = k_0 \land \tilde{p}(1) = k_1.
            \end{split}
        \end{equation}
    \end{itemize}
\end{itemize}

\end{contract}
\begin{contract}[$\mathcal{C}_{\Planner}$] The contract for the $\Planner$ is as follows:
\begin{itemize}
    \small
    \item Assumes
    \begin{itemize}
        \item When the $\Tracker$ completes its task according to the corridor map $\delta$, it should send a report to the \Planner:
\begin{equation}
\label{planner:assume:tracking_complete}
    \begin{split}
        &\forall c \in \mathbf{C} :: \exists p \in \mathbf{P} :: \forall p' \in \mathbf{P}:: \\ &[p', c] \in \text{send}_{\Planner_{\rightarrow \Tracker}} \preceq [p, c] \in \text{send}_{\Planner_{\rightarrow \Tracker}} \land \\ 
        &\textit{c.car.state} \in \Gamma_{\delta}(p, 1) \\ 
        &\leadsto [\texttt{Completed}, c] \in \text{receive}_{\Planner_{\leftarrow \Tracker}}.
    \end{split}
\end{equation}
        \item If the $\Tracker$ sees a failure, it should report to the $\Planner$:
\begin{equation}
\label{planner:assume:car_broke}
    \begin{split}
        &\forall c \in \mathbf{C} :: \lnot \textit{c.car.healthy} \\ &\leadsto
        [\texttt{Failed}, c] \in \text{receive}_{\Planner_{\leftarrow \Tracker}}.
    \end{split}
\end{equation}
    \end{itemize}
    \item Guarantees
    \begin{itemize}
        \item When receiving a pair of configurations from the $\Supervisor$, the $\Planner$ should send a path to the $\Tracker$ such that the starting and ending configurations of the path match the received configurations or if this is not possible, send \texttt{Blocked} to the \Supervisor:
        \begin{equation}
        \label{planner:guarantee:convert_configuration_to_path}
            \begin{split}
                &\forall c \in \mathbf{C} :: \exists (p_0, p_1) \in \mathbb{R}^6 ::\\ &[(p_0, p_1),c] \in \text{receive}_{\Planner_{\leftarrow \Supervisor}} \\ &\leadsto (\exists p  \in \mathbf{P} :: \tilde{p}(0) = p_0 \land \tilde{p}(1) = p_1 \\ &\land [p, c] \in \text{send}_{\Planner_{\rightarrow\Tracker}} \\&
                \lor [\texttt{Blocked}, c] \in \texttt{send}_{\Planner_{\rightarrow \Supervisor}}).
            \end{split}
        \end{equation}
    
        \item Only send safe paths with $\kappa$-feasible curvature:
\begin{equation}
\label{planner:guarantee:trackable_paths}
    \begin{split}
        &\forall c \in \mathbf{C} :: \exists p \in \mathbf{P} :: [p, c] \in \text{send}_{\Planner_{\rightarrow \Tracker}} \\
        &\Rightarrow\text{$\kappa$-feasible}(p) \land \Gamma_{\delta}(p) \subseteq \mathbf{G}.\textit{drivable\_area}.
    \end{split}
\end{equation}
        
        \item If receiving a task status update from the $\Tracker$, eventually forward it to the $\Supervisor$:
\begin{equation}
\label{planner:guarantee:forward_updates}
    \begin{split}
        &\forall c \in \mathbf{C} :: [m, c] \in \text{receive}_{\Planner_{\leftarrow \Tracker}} \\ 
        &\land m \in \{\texttt{Failed}, \texttt{Completed} \}
        \leadsto \\
        &[m, c] \in \text{send}_{\Planner_{\rightarrow \Supervisor}}.
    \end{split}
\end{equation}
        \item If the $\Planner$ receives a \texttt{Blocked} signal from the $\Tracker$, it attempts to fix it, otherwise forwards it to the $\Supervisor$:
\begin{equation}
\label{planner:guarantee:blockage_fix_or_report}
    \begin{split}
        &\forall c \in \mathbf{C}:: \forall t :: [\texttt{Blocked}, c] \in \text{receive}_{\Planner_{\leftarrow \Tracker}}(t) \Rightarrow \\
        &\exists \varepsilon \in \mathbf{R}_{\geq 0} :: (\exists p \in \mathbf{P} :: [p, c] \in  \text{send}_{\Planner_{\rightarrow \Tracker}}(t+\varepsilon) \\
        &\land \forall t' < t+\varepsilon:: [p, c] \not\in \text{send}_{\Planner_{\rightarrow \Tracker}}(t') \\ 
        &\lor [\texttt{Blocked}, c] \in \text{send}_{\Planner_{\rightarrow \Supervisor}}(t+\varepsilon)).
    \end{split}
\end{equation}
    \end{itemize}
\end{itemize}
\end{contract}
\begin{contract}[$\mathcal{C}_{\Tracker}$] The contract for the tracking component is as follows:
\begin{itemize}
    \small
    \item Assumes
    \begin{itemize}
        \item Any path command from the $\Planner$ is always $\kappa$-feasible, the corresponding corridor is drivable, and the car configuration upon receiving the command is in the initial portion of the corridor:
\begin{equation}
\label{tracker:assume:feasible_path_from_planner}
    \begin{split}
        &\forall c \in \mathbf{C} ::  \forall p \in \mathbf{P} :: \forall t::  \mathrm{starts\_at}([p, c]\\ &\in \text{receive}_{\Tracker_{\leftarrow \Planner}}, t) \land \\ &\textit{$\kappa$-feasible}(p) \land \textit{c.car.state}(t) \in \Gamma_{\delta}(p, 0) \land \Gamma_{\delta}(p) \\ &\subseteq \mathbf{G}.\textit{drivable\_area}.
    \end{split}
\end{equation}
        \item Commands are not modified by the $\Customer$:
    \begin{equation}
    \label{tracker:assume:customer_does_not_modify_inputs}
    \text{See } \eqref{customer:guarantee:no_tampering} \text{ and } \eqref{customer:guarantee:no_tracker_input}.
    \end{equation}
    \end{itemize}
    \item Guarantees
    \begin{itemize}
        \item Make sure car stays in the latest sent $p$'s corridor  $\Gamma_{\delta}(p)$:
        \begin{equation}
        \label{tracker:guarantee:stay_in_corridor}
            \begin{split}
            &\forall c \in \mathbf{C} ::
            \forall t ::\exists p \in \mathbf{P} :: \forall p' \in \mathbf{P} :: \\
            &(([p,c] \in \text{receive}_{\Tracker_{\leftarrow \Planner}}(t) \\& \land [p',c] \in \text{receive}_{\Tracker_{\leftarrow \Planner}}(t)) \\
            &\Rightarrow([p',c] \in \text{receive}_{\Tracker_{\leftarrow \Planner}} \\ &\preceq
           [p,c] \in \text{receive}_{\Tracker_{\leftarrow \Planner}})) \\
           &\Rightarrow\textit{c.car.state}(t) \in \Gamma_{\delta}(p). 
            \end{split}
        \end{equation}
        
        \item Tracking command inputs are compatible with cars:
        \begin{equation}
         \label{tracker:guarantee:compatible_tracking_inputs}
            \begin{split}
                &\forall c \in \mathbf{C}:: \square_{\geq0} ([(v, \varphi), c] \in \text{send}_{\Tracker_{\rightarrow \Customer}} \Rightarrow \\
                &v_{\min} \leq v \land v \leq v_{\max} 
        \land \varphi_{\min} \leq \varphi \land \varphi \leq \varphi_{\max}).
            \end{split}
        \end{equation}
        \item Never drive into a dynamic obstacle (customer or car):
        \begin{equation}
         \label{tracker:guarantee:no_collision}
            \begin{split}
                &\forall c_1, c_2 \in \mathbf{C} :: \square_{\geq0} ((c_1 \neq c_2 \Rightarrow \\ &\norm{(c_1.\textit{car.x}, c_1.\textit{car.y}) - (c_2.\textit{car.x}, c_2.\textit{car.y})} \geq \varepsilon_{\min, car}) \\& \land 
                \norm{(c_1.\textit{car.x}, c_1.\textit{car.y}) - (c_2.\textit{x}, c_2.\textit{y})} \geq \varepsilon_{\min, people})
                ).
            \end{split}
        \end{equation}
        \item If a car fails, it must report to the $\Planner$:
        \begin{equation}
         \label{tracker:guarantee:must_report_failure_to_planner}
    \begin{split}
        &\forall c \in \mathbf{C} :: \lnot \textit{c.car.healthy} \\ &\leadsto
        [\texttt{Failed}, c] \in \text{send}_{\Tracker_{\leftarrow \Planner}}.
    \end{split}
\end{equation}
            \item If  a car is healthy then it must ``track'' the last sent path from the \Planner:
            \begin{equation}
            \label{tracker:guarantee:healthy_must_track_last_path}
                \begin{split}
                    &\forall c \in \mathbf{C} :: \square_{\geq0} \textit{c.car.healthy} \land \exists p \in \mathbf{P} :: \forall p' \in \mathbf{P} :: \\
                    &[p', c] \in \text{receive}_{\Tracker_{\leftarrow \Planner}} \preceq [p, c] \in \text{receive}_{\Tracker_{\leftarrow \Planner}} \\ &\Rightarrow \exists t :: \textit{c.car.state}(t) \in \Gamma_{\delta}(p, 1).
                \end{split}
            \end{equation}

            \item When the $\Tracker$ completes its task according to a corridor map $\delta$, it should send a report to the $\Planner$ module:
        \begin{equation}
         \label{tracker:guarantee:when_task_completed_report_to_planner}
            \text{ See } \eqref{planner:assume:tracking_complete}.
        \end{equation}
        \item If a car is blocked (i.e., there is a failed car in its current corridor), then the \Tracker\: must report \texttt{Blocked} to the \Planner:
        \begin{equation}
        \label{tracker:guarantee:if_blocked_must_report}
            \begin{split}
                &\forall c \in \mathbf{C} :: \forall t :: \exists p \in \mathbf{P} :: ([p, c] \text{receive}_{\Tracker_{\leftarrow \Planner}}(t) :: \\ &\forall p' \in \mathbf{P} :: [p',c] \in \text{receive}_{\Tracker_{\leftarrow \Planner}}(t') \\ &\Rightarrow [p', c] \in \text{receive}_{\Tracker_{\leftarrow \Planner}} \\&\preceq [p, c] \in \text{receive}_{\Tracker_{\leftarrow \Planner}}) \\
                &\Rightarrow(\exists c' \in \mathbf{C} :: c' \neq c \\ &\land \lnot \textit{c'.car.healthy} \land \textit{c'.car.state}(t) \in \Gamma_{\delta}(p) \\ &\leadsto [\texttt{Blocked}, c] \in \text{send}_{\Tracker_{\rightarrow \Planner}}).
            \end{split}
        \end{equation}
        \end{itemize}
\end{itemize}
\end{contract}

\section{System design}
\begin{figure}
      \centering
      \includegraphics[width=\linewidth*3/4]{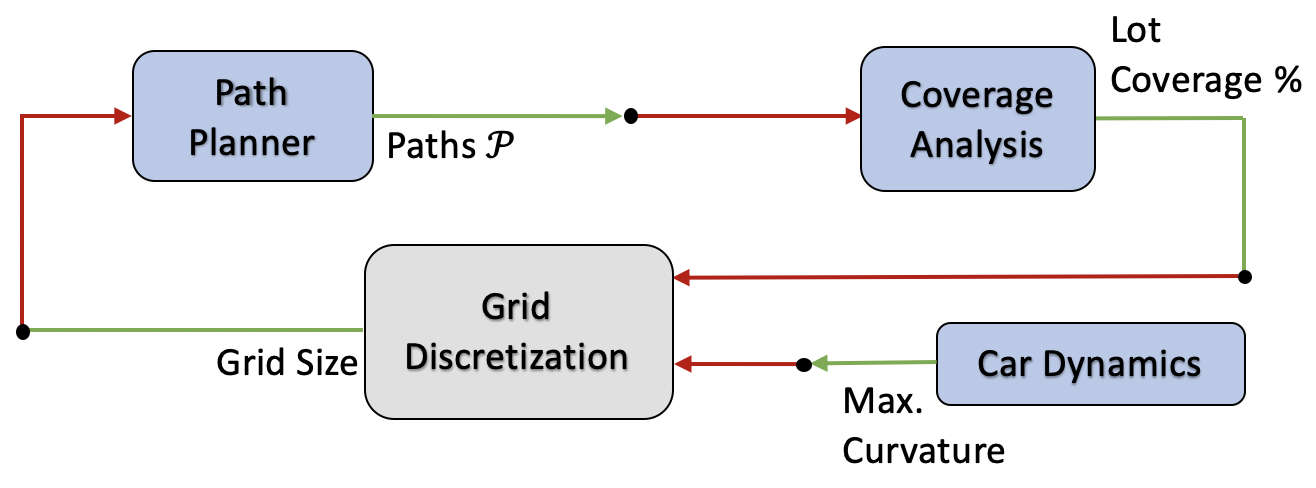}
      \caption{Implementation of the Planner component.}
      \label{fig:plannerimpl}
\end{figure}
\subsection{Simulation Environment and Implementation}
The proposed design framework was demonstrated via simulation of an automated valet parking (AVP) system~\cite{avpsim_video}. It consists of the layout of a parking lot (Fig. \ref{fig:lot}), as well as multiple cars that arrive at the drop off location of the parking lot and are parked in one of the vacant spots by the AVP system. Once the customer requests their car, it is returned to the pick-up location. The asynchronicity is captured by modeling each component as a concurrent process using Python \texttt{async} library \texttt{Trio} \cite{trio}.
The communication between the layers is implemented using \texttt{Trio}'s \texttt{memory\_channel}. In particular, each channel is a first-in-first-out queue which ensures losslessness. The architecture is described in Figure \ref{fig:contracts}. 
In this setup, the cars may experience failures and report them to the \Tracker~module. The failures considered in this demonstration are a blocked path, a blocked parking spot, and a total engine failure resulting in immobilization.
The benefit of the directive-response architecture becomes apparent when failures are introduced into the system. 
Upon experiencing a failure, a component that is higher in the hierarchy will be alerted through the response it receives. If possible, the failure will be resolved, e.g., through the re-planning of the path or assigning a different spot. 
Every layer has access to its contingency plan, consisting of several predetermined actions according to the possible failure scenarios and corresponding responses it receives.
In some cases (e.g., complete blockage of a car), when no action can resolve the issue, the cars have to wait until the obstruction is removed. We assume that only broken cars can be towed, and when a car breaks down, it will take a specified amount of time until it is towed.
\subsection{\Customer~Modeling}
In our simulation, customers are responsible for driving their cars into the parking garage and depositing them at the drop-off area with an admissible configuration before sending a \texttt{Park} directive to the $\Supervisor$ and stay there until they get a response. This is satisfied as long as the customer drops off their vehicle behind the green line such that the heading of the vehicle is within the angle bounds $\underline{\alpha}$ and $\overline{\alpha}$ as shown in Figure~\ref{fig:initial_curvature_bound} with the projection $w$ of the vehicle onto the green edge of the blown-up entrance box shown in Figure~\ref{fig:expath}. Therefore, $\Customer$ satisfies $G_{\eqref{customer:guarantee:stay_in_box}}$. If the \texttt{Park} directive is \texttt{Rejected} by the $\Supervisor$, the customer is assumed to be able to leave the garage safely (satisfying $G_{\eqref{customer:guarantee:rejected_leave_forever}}$). If the car is \texttt{Accepted} then the customer will leave the control of the car to the $\Tracker$ (satisfying $G_{\eqref{customer:guarantee:no_tampering}}$ and $G_{\eqref{customer:guarantee:no_tracker_input}}$). The customer is assumed to always eventually send a \texttt{Retrieve} directive to the $\Supervisor$, after their car is \texttt{Accepted} (satisfying $G_{\eqref{customer:guarantee:will_pick_up}}$ and $G_{\eqref{customer:guarantee:must_receive_accepted_before_retrieve}}$). Once the vehicle is \texttt{Returned}, the customer is assumed to be able to pick it up and drive safely away. All pedestrians in the parking lot are customers, and they are constrained to only walk on the walkable area and never stay on a crosswalk forever (thus satisfying $G_{\eqref{customer:guarantee:on_walkable_area}}$ and $G_{\eqref{customer:guarantee:transient_crosswalks}}$). When a car fails, it becomes immobilized until it is towed ($G_{\eqref{customer:guarantee:unhealthy_cannot_move}}$).
From this, it follows that \Customer\: satisfies $\mathcal{C}_{\Customer}$.
\subsection{\Supervisor~Implementation}
At any time, the $\Supervisor$ knows the total number of cars that have been accepted into the garage, which is represented by the variable $num\_active\_customers$, and is designed to accept new cars when this number is strictly less than the total number of parking spots $\mathbf{G}.parking\_spots$. This implies $G_{\eqref{supervisor:guarantee:acceptance}}$ is satisfied. Overall, this ensures all directives will get a response, yielding $G_{\eqref{supervisor:guarantee:responsive}}$. Whenever the $\Supervisor$ receives a \texttt{Completed} signal, it will check if the car is the return area. If it is, then the $\Supervisor$ will send a $\texttt{Returned}$ signal to the $\Customer$ in compliance with $G_{\eqref{supervisor:guarantee:only_send_returned_when_really_returned}}$. If the $\Supervisor$ ever accepts a new car, or receives a \texttt{Blocked} signal from the $\Planner$, or a \texttt{Retrieve} request it will send a start configuration compatible with the car's current state as well as an end configuration to one of the parking spaces in the former case and to a place in the return area in the latter. This guarantees  $G_{\eqref{supervisor:guarantee:send_pair_of_configurations}}$.
\begin{proposition}
$M_{\Supervisor}$ satisfies $\mathcal{C}_{\Supervisor}.$
\end{proposition}
\begin{proof}

Let $M$ denote our implementation of the \Supervisor\: and $\sigma \in M$. We want to show that $$\sigma \in 
\bigwedge\limits_{i=\ref{supervisor:assume:all_broken_will_be_towed}}^{\ref{supervisor:assume:planner_reports_complete}} A_{(i)}
\Rightarrow \sigma \in \bigwedge\limits_{i=\ref{supervisor:guarantee:responsive}}^{\ref{supervisor:guarantee:send_pair_of_configurations}} G_{(i)}.$$ 

From the description of the $\Supervisor$ implementation, we conclude $\sigma \in G_{\eqref{supervisor:guarantee:responsive}} \land
G_{\eqref{supervisor:guarantee:only_send_returned_when_really_returned}} \land G_{\eqref{supervisor:guarantee:healthy_retrieval_implies_last_command_is_to_return_area}} \land
G_{\eqref{supervisor:guarantee:acceptance}} \land 
G_{\eqref{supervisor:guarantee:send_pair_of_configurations}}$. Since $\sigma \in A_{\eqref{supervisor:assume:planner_reports_complete}}$ and because in our implementation whenever the $\Supervisor$ receives a \texttt{Completed} signal it will alert the customer of the corresponding status, our implementation satisfies $G_{\eqref{supervisor:guarantee:eventual_return}}$.
\end{proof}

\subsection{\Planner~Implementation}
The $\Planner$ computes paths that cover the parking spots, as well as the entry and exit areas of the parking garage, which are $\kappa$-feasible for a car that satisfies $\eqref{dubins}$ such that the corresponding $\delta$-corridor is on $\mathbf{G}.drivable\_area$. 
Given a maximum allowable curvature, a grid discretization scheme is based on a planning grid whose size is computed to provide full lot coverage and satisfy the curvature bounds, as depicted in Figure~\ref{fig:plannerimpl}. 
For every specified grid size, the algorithm will check if the planning graph is appropriate by determining how well the parking lot is covered.
Only a grid size that provides full coverage of the lot is chosen for path planning.
The dynamical system specified in~\eqref{dubins} is differentially flat \cite{fliess1995flatness}. In particular, it is possible to compute all states and inputs to the system, given the outputs $x, y$, and their (in this case, up to second order) derivatives. Specifically, the steering input is given by
\begin{equation}
\label{phi-kappa}
    \varphi(t) = \arctan(\ell \kappa(t)), 
\end{equation}
where $\kappa(t)$ is the curvature of the path traced by the midpoint of the rear axle at time $t$ given by
\begin{equation}
    \kappa(t) = \frac{\ddot{y}(t)\dot{x}(t)-\ddot{x}(t)\dot{y}(t)}{\big({\dot{x}^2(t) + \dot{y}^2(t)}\big)^{\frac{3}{2}}}.
\end{equation}
The task of tracking a given path can be shown to
depend only on how $\varphi(t)$ is constrained. For practical purposes, let us assume $\lvert{\varphi(t)\rvert} \leq B$ for some $B > 0$. Then by Equation~\eqref{phi-kappa}, tracking feasibility depends on whether the maximum curvature of that path exceeds $\frac{\tan(B)}{\ell}$. For our implementation this is assumed to be $0.2\;m^{-1}$.
This problem has been studied in~\cite{cowlagi2011hierarchical} in the context of rectangular cell planning. We apply the algorithm described therein for a Type 1 path (CBTA-S1) to a rectangular cell while constraining the exit configuration to a heading difference of $\pm5^\circ$ and a deviation of $\pm0.5\:m$ from the nominal path. The setup and the resulting initial configuration, for which traversal is guaranteed, are shown in Figure~\ref{fig:initial_curvature_bound} and  Figure~\ref{fig:expath}. The initial car configuration can be anywhere on the grid segment entry edge, as long as it is between the lower bound $\underline{\alpha}$ and the upper bound $\overline{\alpha}$. By passing through this initial funnel segment, the car will transition itself onto the planning grid. Therefore, it remains to be verified that each path generated from the grid is guaranteed to have a maximum curvature that is smaller than $\kappa$. An example path and its curvature are provided in Figure~\ref{fig:expath}. Combining the parking lot coverage, initial grid segment traversability, and the curvature analysis, a grid size is determined to be $3.0\:m$ for the path planner, according to Figure~\ref{fig:plannerimpl}. The synthesized grid size and path smoothing technique used in our $\Planner$ guarantee that all trajectories generated meet this maximum curvature requirement. In addition to satisfying $G_{\eqref{planner:guarantee:forward_updates}}$, any execution of the $\Planner$ also satisfies  $G_{\eqref{planner:guarantee:convert_configuration_to_path}}$ and $G_{\eqref{planner:guarantee:trackable_paths}}$ because either the $\Planner$ can generate a feasible path or it will send a \texttt{Blocked} signal to the $\Supervisor$. When the $\Planner$ receives a \texttt{Blocked} signal from the $\Tracker$ it will either attempt to find a different path on the planning graph or report this to the $\Supervisor$. This satisfies $G_{\eqref{planner:guarantee:blockage_fix_or_report}}$.
\begin{figure}
        \centering
        \includegraphics[width=\linewidth*3/4]{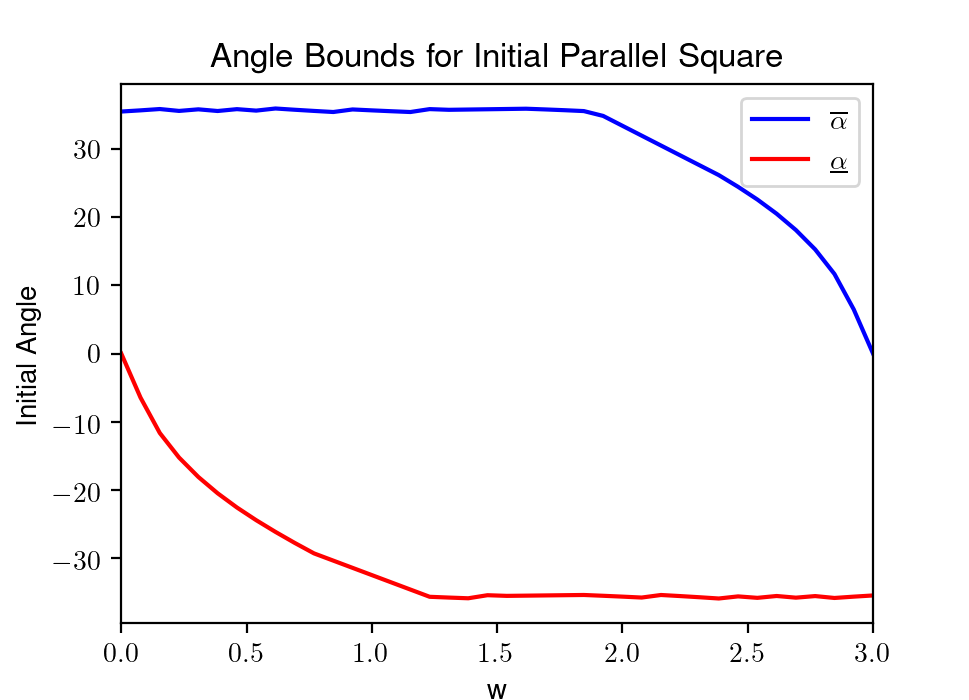}
        \caption{Possible initial car configuration along the entrance region (green) corresponding to a grid square as defined in Fig.~\ref{fig:expath}.}
        \label{fig:initial_curvature_bound}
\end{figure}

\begin{figure}
      \centering
      \includegraphics[width=\linewidth*3/4]{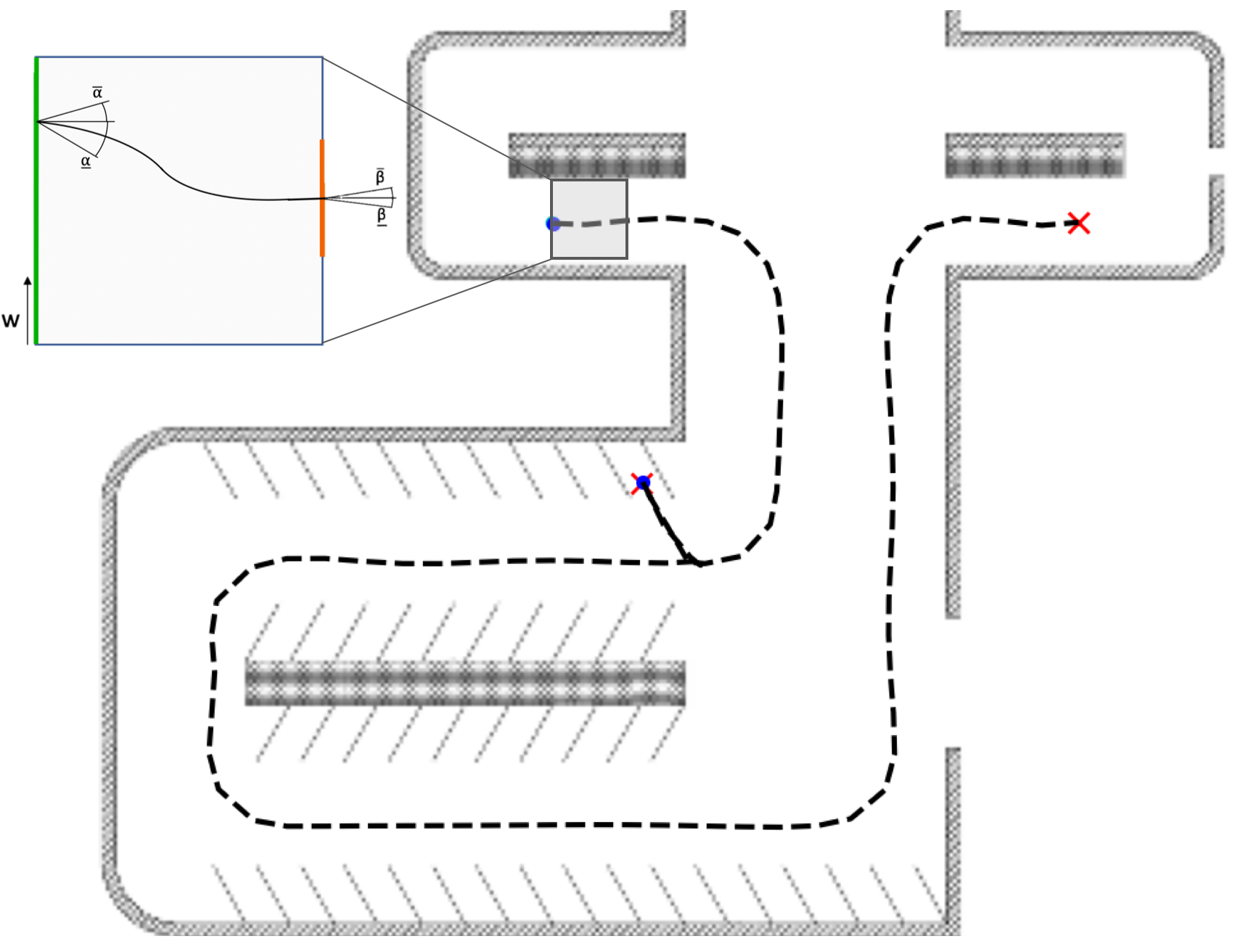}
      \includegraphics[width=\linewidth]{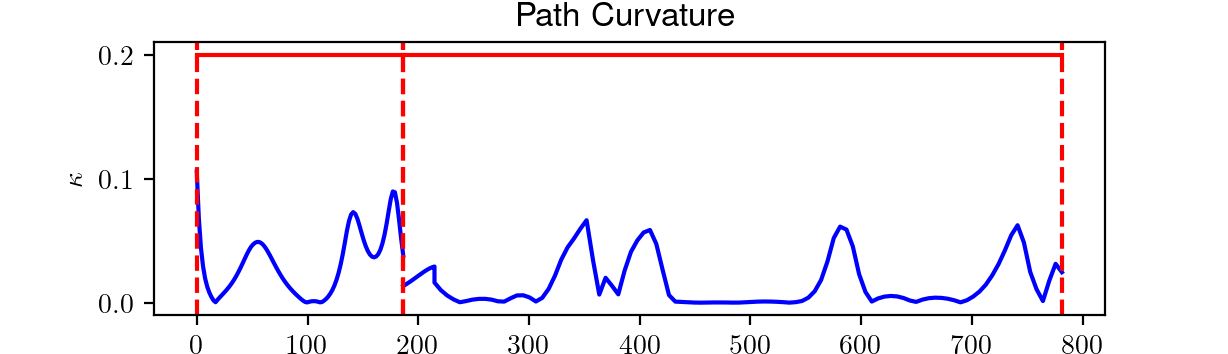}
      \caption{Example path through the parking lot and corresponding curvature and initial grid segment layout.}
      \label{fig:expath}
\end{figure}

\subsection{\Tracker~Implementation} 
The $\Tracker$ receives directives from the $\Planner$ consisting of trackable paths and sends responses according to the task status to the $\Planner$.
The $\Tracker$ sees all agents in $\mathbf{G}.interior$ and guarantees no collisions by sending a brake signal when necessary to ensure a minimum safe distance is maintained at all times. 
The tracking algorithm that we use is an off-the-shelf MPC algorithm from \cite{sakai2018pythonrobotics}.

To ensure that the vehicles stay in the $\delta$-corridors, given knowledge of the vehicle's dynamics, we can
synthesize motion primitives that are robust to a certain disturbance set $\Delta_{Car}$ (see Figure~\ref{fig:contracts}). Algorithms for achieving this have been proposed and implemented, for example, in \cite{schurmann2017guaranteeing} for nonlinear, continuous-time systems and for affine, discrete-time systems in \cite{filippidis2016control}. In our implementation we used a simplified approach, which ensures that a backup controller for the car gets activated if the car approaches the boundary of the $\delta$-corridor and ensures that the car will merge onto the path again. Once it reaches the original path, the tracking of the remaining path will continue.

By $A_{\eqref{tracker:assume:feasible_path_from_planner}}$, any new path command $[p, c]$ sent down from the \Planner\: module is assumed to be $\kappa$-feasible and have a drivable $\delta$-corridor, the initial portion of which contains $c.car$ at that time. In our implementation, we ensure that every time this happens, $c.car$ is stationary. And under this condition, we were able to confirm by testing that a car controlled by the MPC algorithm can track the corresponding $\delta$-corridors of a diverse enough set of paths, thus satisfying $G_{\eqref{tracker:guarantee:stay_in_corridor}}$ and $G_{\eqref{tracker:guarantee:healthy_must_track_last_path}}$. The MPC algorithm is configured to output properly bounded control inputs, thus satisfies $G_{\eqref{tracker:guarantee:compatible_tracking_inputs}}$. In addition, our implementation satisfies $G_{\eqref{tracker:guarantee:must_report_failure_to_planner}}$, $G_{\eqref{tracker:guarantee:when_task_completed_report_to_planner}}$ and $G_{\eqref{tracker:guarantee:if_blocked_must_report}}$ by construction.
And finally, we can guarantee $G_{\eqref{tracker:guarantee:no_collision}}$ by Property~\ref{property:safety}.

\section{Correctness of the Composed System}
\label{proofs}
In this section, we will show that our implementation of the AVP is correct and satisfies the overall system specification by leveraging the modularity provided by the contract based design. 
We start by composing the AVP components, namely the \Supervisor, the \Planner, and the \Tracker\: and then computing the quotient of the overall specification and the composed contract. Then we will show that our contract for the \Customer\: is a refinement of this quotient.
\subsection{Contract Composition}
As part of the final verification step, we will be taking the composition of the component contracts and showing that our overall system implementation satisfies this composition. This will imply that the composition is consistent. 

Given two saturated contracts $\mathcal{C}_1$ and $\mathcal{C}_2$, their composition $\mathcal{C}_1 \otimes \mathcal{C}_2 =(A,G)$ given by~\cite{benveniste2018contracts}:
\begin{equation*}
\label{eq:comp_ag}
G=G_{1}\land G_{2} \text{ and }A = \bigvee_{A' \in \mathcal{A}} A',
\end{equation*}
where
\begin{equation*}
\mathcal{A}={\left\{\begin{array}{@{}c|c@{}}
 & A' \land  G_{2} \Rightarrow A_{1}\\A' & \texttt{and}\\ & A' \land  G_{1}\Rightarrow A_{2}.
\end{array}\right\}}
\end{equation*}
A nice property of the composed contract is that if $M_1$ satisfies $\mathcal{C}_1$ and $M_2$ satisfies $\mathcal{C}_2$ then $M_1 \times M_2$ satisfies $\mathcal{C}_1 \otimes \mathcal{C}_2$. Using the fact that the composition operator $\otimes$ is associative and commutative, a straightforward calculation yields the following more explicit form for the composition of $N$ saturated contracts $(A_i, G_i)_{i=1}^N$.
$$G = \bigwedge_{i=1}^{N} G_i \text{ and }
A = \bigwedge_{i=1}^{N} A_i  \lor \lnot \Big( \bigwedge_{i=1}^{N} G_i \Big)$$
If $A \neq \varnothing$, then the composed contract is compatible. The contract is consistent if there exists an implementation for it, namely $G \neq \varnothing$ if it is saturated. For our AVP system, we will show that our composed implementation also satisfies the composed contract in a non-vacuous way, meaning it satisfies all guarantees of the component contracts simultaneously.
In the composition, an acceptable behavior satisfies the following properties, namely the operation of the car inside the garage:
\begin{enumerate}
    \item The \Supervisor\: rejects the car due to the lack of reachable, vacant spots. The car will not enter the garage.
    \item A car which was dropped off correctly in the deposit area is accepted by the \Supervisor\: by $G_{\eqref{supervisor:guarantee:eventual_return}}$.
    \begin{enumerate}
        \item Accepted, no contingency: The $\Tracker$ by $A_{\eqref{tracker:assume:customer_does_not_modify_inputs}}$ takes over control. After this, the $\Supervisor$ must send a directive in the form of a pair of configurations to the $\Planner$ $G_{\eqref{supervisor:guarantee:send_pair_of_configurations}}$, which in turn must send to the $\Tracker$ a safe and feasible path (satisfying $A_{\eqref{tracker:assume:feasible_path_from_planner}}$) such that the starting and ending configurations of the path match the received configurations ($G_{\eqref{planner:guarantee:convert_configuration_to_path}}$, $G_{\eqref{planner:guarantee:trackable_paths}}$). Upon receiving the path from the $\Planner$, the $\Tracker$ ensures that the car stays in the corridor of the path $G_{\eqref{tracker:guarantee:stay_in_corridor}}$ and ensures that it will make progress on that path (this satisfies $G_{\eqref{tracker:guarantee:healthy_must_track_last_path}}$). It will accomplish this while sending compatible inputs to the customer's car $G_{\eqref{tracker:guarantee:compatible_tracking_inputs}}$ and not driving it into people and other cars $G_{\eqref{tracker:guarantee:no_collision}}$. When the $\Customer$ sends a \texttt{Retrieve} command, the above process repeats with the $\Supervisor$, which ensures that the last configuration is in the return area, thus satisfying $G_{\eqref{supervisor:guarantee:healthy_retrieval_implies_last_command_is_to_return_area}}$.
        If this is the last sent path, then upon reaching the end of the path, it should notify the $\Planner$ module that it has completed the task by $G_{\eqref{tracker:guarantee:when_task_completed_report_to_planner}}$ which satisfies $A_{\eqref{supervisor:assume:planner_reports_complete}}$, $A_{\eqref{supervisor:guarantee:acceptance}}$, and $A_{\eqref{planner:assume:tracking_complete}}$. The $\Supervisor$ alerts the $\Customer$ of the completed return by $G_{\eqref{supervisor:guarantee:only_send_returned_when_really_returned}}$.
    \end{enumerate}
    \item Accepted, with problems: If the car is accepted and at any time during the above process: 
        \begin{enumerate}
            \item The car fails (hence, cannot move by $G_{\eqref{customer:guarantee:unhealthy_cannot_move}}$), the $\Tracker$ will send a \texttt{Failed} message to the $\Planner$ by $G_{\eqref{tracker:guarantee:must_report_failure_to_planner}}$ satisfying $A_{\eqref{planner:assume:car_broke}}$ and by $G_{\eqref{planner:guarantee:forward_updates}}$ this will be forwarded to the $\Supervisor$. This satisfies $A_{\eqref{supervisor:assume:fail_report}}$, which together with $A_{\eqref{supervisor:assume:all_broken_will_be_towed}}$, will imply that the failed car will eventually be towed.
            \item The car is \texttt{Blocked}, the $\Tracker$ will report to the $\Planner$ by $G_{\eqref{tracker:guarantee:if_blocked_must_report}}$, which will try to resolve or alert the $\Supervisor$ satisfying $G_{\eqref{planner:guarantee:blockage_fix_or_report}}$.
        \end{enumerate}
\end{enumerate}
\subsection{Contract Quotient}
For saturated contracts $\mathcal{C}$ and $\mathcal{C}_1$, the quotient is defined in \cite{romeo2018quotient} as follows:
\begin{equation}
    \mathcal{C}/\mathcal{C}_1 :=(A\land G_1,A_1 \land G\lor \lnot(A\land G_1))
\end{equation}
Quotienting out the composed specification of the components from the overall system specification should yield the the required customer behavior. The composed system was computed to be $$\mathcal{C}_{AVP}=((A_S\land A_P\land A_T)\lor \neg(G_S\land G_P\land G_T), G_S \land G_P \land G_T)$$ with the assertions $G_i$ and $A_i$ of the Supervisor, Planner and Tracker contract $\mathcal{C}_i$
in saturated form.
 With the contract for the overall system defined as:
\begin{contract}[$\mathcal{C}_{System}$] The contract for the overall system is as follows:
\begin{itemize}
    \small
    \item Assumes
    \begin{itemize}
        \item Any circumstances.
    \end{itemize}
    \item Guarantees
    \begin{itemize}
        \item Never any collisions (safety).
        \item Always healthy cars will eventually be returned (liveness).
    \end{itemize}
\end{itemize}
 \end{contract}
 When computing the quotient of the overall system specification and the composed system, the resulting assumptions and guarantees are the following.
Assuming that the AVP components work correctly (e.g. provide their respective guarantees), the customer must guarantee that all assumptions that the AVP components make on the customer are valid, while ensuring safety and progress. Meaning the customer need to provide the following guarantees:
\begin{itemize}
    \item Guarantees:
    \begin{itemize}
        \item The customer will drop off the car correctly satisfying $A_{\eqref{supervisor:assume:cowlagi}}$.
        \item The customer will not interfere with the car controls after the drop-off satisfying $A_{\eqref{tracker:assume:customer_does_not_modify_inputs}}$.
        \item The customer needs to ensure progress by not blocking the path forever, and eventually requesting and picking up the car.
        \item The customer will not take any action towards collision ensuring safety.
    \end{itemize}
\end{itemize}
Our customer contract refines the contract with the above mentioned guarantees. $A_{\eqref{supervisor:assume:cowlagi}}$ and $A_{\eqref{tracker:assume:customer_does_not_modify_inputs}}$ are satisfied by $G_{\eqref{customer:guarantee:stay_in_box}}$ and $G_{\eqref{customer:guarantee:no_tampering}}$. The safety property is guaranteed by the customer staying in the walkable area by $G_{\eqref{customer:guarantee:on_walkable_area}}$. Progress is ensured by $G_{\eqref{customer:guarantee:will_pick_up}}$, $G_{\eqref{customer:guarantee:transient_crosswalks}}$, and $G_{\eqref{customer:guarantee:rejected_leave_forever}}$. Our \Customer\: contract includes the guarantees generated from the quotient and thus is a refinement of this contract.

We will now show specifically that the composed system satisfies the safety and progress properties ($G_{\eqref{tracker:guarantee:no_collision}}$ and $G_{\eqref{supervisor:guarantee:eventual_return}}$):
\begin{property}[Safety]
\begin{equation}
\label{property:safety}
            \begin{split}
                &\forall c_1, c_2 \in \mathbf{C} :: \square ((c_1 \neq c_2 \Rightarrow \\ &\norm{(c_1.\textit{car.x}, c_1.\textit{car.y}) - (c_2.\textit{car.x}, c_2.\textit{car.y})} \geq \varepsilon_{\min, car}) \land \\ 
                &\norm{(c_1.\textit{car.\textit{x}}, c_1.\textit{car.y}) - (c_2\textit{.x}, c_2.\textit{y})} \geq \varepsilon_{\min, people})
                ).
            \end{split}
        \end{equation}\end{property}
\begin{proof}{(Sketch)}
For each vehicle in the parking lot, the following invariance is maintained.
There will be no collisions, as the $\Tracker$ checks the spatial region in front of the car and brings it to a full stop in case the path is blocked by another agent (car or pedestrian). The minimum distance to an obstacle is determined by a minimum braking distance.
Furthermore, the environment does not take actions, which will lead to an inevitable collision due to the constraints on the pedestrian dynamics~\ref{control_limits}.
\end{proof}
\begin{property}[Liveness]
\label{property:liveness}
\begin{equation}
        \begin{split}
   &\forall c \in \mathbf{C} :: (\square \textit{c.car.healthy} \land \lozenge\square \neg [\texttt{Blocked},c]\\ &\in \text{receive}_{\Supervisor
   -{\leftarrow\Planner}}
   \Rightarrow \\
                &[\texttt{Retrieve}, c] \in \text{receive}_{\Supervisor_{\leftarrow \Customer}} \leadsto \\
                &[\texttt{Returned}, c] \in \text{receive}_{\Customer_{\leftarrow \Supervisor}}
                ).        \end{split}
    \end{equation}
\end{property}
\vspace{2mm}
\begin{proof}{(Sketch)}
Consider the parking lot topology shown in Figure~\ref{fig:lot}.
Let $c \in \mathbf{C}$ and $\textit{c.car.healthy}$. Assume that $c$ sends a \texttt{Retrieve} message to the \Supervisor\:. For each $t$, let us define $f(t)$ to be the number cars between $\textit{c.car}$ and its destination. Clearly, $f(t) \geq 0$ for any $t$ and  $f(t)$ is well-defined because for the topology being considered, we can trace out a line that starts from the entrance area, going to any one of the parking spots and ending at the return area without having to retrace our steps at any time. We will show that there exists a $t' \geq t$ such that $f(t') = 0$, implying that there is no longer any obstacle between $c$ and its destination. Next, we claim that $\forall t,t'::t' > t:: f(t) \geq f(t')$. This is true because:
\begin{itemize}
    \item The parking lot topology and the safety measures do not allow for overtaking.
    \item The area reservation strategy implemented in the \Supervisor\: prevents an increase in $f$ upon re-routing to avoid a failed car. A notable detail is that if $\textit{c.car}$ is trying to back out of a parking spot, a stream of cars passing by can potentially block it forever. This is resolved by having $\textit{c.car}$ reserve the required area so that once any other car has cleared this area, $\textit{c.car}$ is the only one that has the right to enter it.
\end{itemize}
Finally, we will show that $\forall t:: \exists t':: t' > t:: f(t) > f(t')$.
 Let $c'$ be such that $c'.car$ is between $c.car$ and its destination. By the dynamical constraint on pedestrians and by assumptions $A_{\eqref{customer:guarantee:on_walkable_area}}$ and $A_{\eqref{customer:guarantee:transient_crosswalks}}$, they will not block cars forever. Our algorithm guarantees that one of the following will happen at some time $t' > t$:
    \begin{enumerate}
        \item $c'.car$ is picked up by $c'$.
        \item $c'.car$ is parked and $c.car$ drives past it
        \item $c'.car$ drives past $c.car$'s destination.
        \item $c'.car$ breaks down and by $A_{\eqref{supervisor:assume:all_broken_will_be_towed}}$ is eventually towed.
    \end{enumerate}
It is easy to see that each of these events implies that $f(t) > f(t')$. Since $f$ is an integer and cannot drop below $0$, the result follows.
\end{proof}

\addtolength{\textheight}{-3cm}   

\section{SUMMARY AND FUTURE WORK}
We have formalized an assume-guarantee contract variant with communication via a directive-response framework. We then used it to write specifications and verified the correctness of an AVP system implementation~\cite{avpsim_video}. This was done separately for each module and everything together as a complete system. 

The application of this framework in the AVP can be extended to more agent types, for example, human-driven cars and pedestrians that do not necessarily follow traffic rules at all times. A contract between the valet driven cars and the human-driven cars will be needed to ensure the safe operation of the parking lot, and in the event that a human-driven car violates the contract, cars controlled by the system need to be able to react to this situation safely.
More failure scenarios such as communication errors (message loss, cyberphysical attacks etc.) may also be included.

\section{ACKNOWLEDGMENTS}

This research was supported by DENSO International America, Inc and National Science Foundation award CNS-1932091.


\bibliographystyle{ieeetr}
\bibliography{refs}

\begin{thebibliography}{10}

\bibitem{benveniste2018contracts}
A.~Benveniste, B.~Caillaud, D.~Nickovic, R.~Passerone, J.-B. Raclet,
  P.~Reinkemeier, A.~L. Sangiovanni-Vincentelli, W.~Damm, T.~A. Henzinger,
  K.~G. Larsen, {\em et~al.}, ``Contracts for system design,'' {\em Foundations
  and Trends in Electronic Design Automation}, vol.~12, no.~2-3, pp.~124--400,
  2018.

\bibitem{censi2015mathematical}
A.~Censi, ``A mathematical theory of co-design,'' {\em arXiv preprint
  arXiv:1512.08055}, 2015.

\bibitem{filippidis2019decomposing}
I.~Filippidis, {\em Decomposing formal specifications into assume-guarantee
  contracts for hierarchical system design}.
\newblock PhD thesis, California Institute of Technology, 2019.

\bibitem{nuzzo2015platform}
P.~Nuzzo, A.~L. Sangiovanni-Vincentelli, D.~Bresolin, L.~Geretti, and T.~Villa,
  ``A platform-based design methodology with contracts and related tools for
  the design of cyber-physical systems,'' {\em Proceedings of the IEEE},
  vol.~103, no.~11, pp.~2104--2132, 2015.

\bibitem{sangiovanni2012taming}
A.~Sangiovanni-Vincentelli, W.~Damm, and R.~Passerone, ``Taming dr.
  frankenstein: Contract-based design for cyber-physical systems,'' {\em
  European journal of control}, vol.~18, no.~3, pp.~217--238, 2012.

\bibitem{damm2011using}
W.~Damm, H.~Hungar, B.~Josko, T.~Peikenkamp, and I.~Stierand, ``Using
  contract-based component specifications for virtual integration testing and
  architecture design,'' in {\em 2011 Design, Automation \& Test in Europe},
  pp.~1--6, IEEE, 2011.

\bibitem{damm2005boosting}
W.~Damm, A.~Votintseva, A.~Metzner, B.~Josko, T.~Peikenkamp, and E.~B{\"o}de,
  ``Boosting re-use of embedded automotive applications through rich
  components,'' {\em Proceedings of Foundations of Interface Technologies},
  2005.

\bibitem{nuzzo2013contract}
P.~Nuzzo, H.~Xu, N.~Ozay, J.~B. Finn, A.~L. Sangiovanni-Vincentelli, R.~M.
  Murray, A.~Donz{\'e}, and S.~A. Seshia, ``A contract-based methodology for
  aircraft electric power system design,'' {\em IEEE Access}, vol.~2,
  pp.~1--25, 2013.

\bibitem{maasoumy2015smart}
M.~Maasoumy, P.~Nuzzo, and A.~Sangiovanni-Vincentelli, ``Smart buildings in the
  smart grid: Contract-based design of an integrated energy management
  system,'' in {\em Cyber Physical Systems Approach to Smart Electric Power
  Grid}, pp.~103--132, Springer, 2015.

\bibitem{BoschAVP}
Bosch, ``Automated valet parking service,'' mar 2020.

\bibitem{BoschAVPDetroit}
Bosch, ``Ford, bedrock and bosch are exploring highly automated vehicle
  technology in detroit to help make parking easier,'' mar 2020.

\bibitem{siemens_avp}
Siemens, ``Improving autonomous valet parking with simulation and testing,''
  mar 2020.

\bibitem{denso_avp}
A.~Yamazaki, Y.~Izumi, K.~Yamane, T.~Nomura, and Y.~Seike, ``Development of
  control technology for controlling automated valet parking,'' mar 2020.

\bibitem{bauer2012moving}
S.~S. Bauer, A.~David, R.~Hennicker, K.~G. Larsen, A.~Legay, U.~Nyman, and
  A.~Wasowski, ``Moving from specifications to contracts in component-based
  design,'' in {\em International Conference on Fundamental Approaches to
  Software Engineering}, pp.~43--58, Springer, 2012.

\bibitem{romeo2018quotient}
{\'I}.~{\'I}. Romeo, A.~Sangiovanni-Vincentelli, C.-W. Lin, and E.~Kang,
  ``Quotient for assume-guarantee contracts,'' in {\em Proceedings of the 16th
  ACM-IEEE International Conference on Formal Methods and Models for System
  Design}, pp.~67--77, IEEE Press, 2018.

\bibitem{wongpiromsarn2008distributed}
T.~Wongpiromsarn and R.~M. Murray, ``Distributed mission and contingency
  management for the darpa urban challenge,'' in {\em International Workshop on
  Intelligent Vehicle Control Systems (IVCS)}, vol.~5, 2008.

\bibitem{dvorak2000software}
D.~Dvorak, R.~Rasmussen, G.~Reeves, and A.~Sacks, ``Software architecture
  themes in jpl's mission data system,'' in {\em 2000 IEEE Aerospace
  Conference. Proceedings}, vol.~7, pp.~259--268, IEEE, 2000.

\bibitem{ingham2005engineering}
M.~D. Ingham, R.~D. Rasmussen, M.~B. Bennett, and A.~C. Moncada, ``Engineering
  complex embedded systems with state analysis and the mission data system,''
  {\em Journal of Aerospace Computing, Information, and Communication}, vol.~2,
  no.~12, pp.~507--536, 2005.

\bibitem{rasmussen2001goal}
R.~D. Rasmussen, ``Goal-based fault tolerance for space systems using the
  mission data system,'' in {\em 2001 IEEE Aerospace Conference Proceedings
  (Cat. No. 01TH8542)}, vol.~5, pp.~2401--2410, IEEE, 2001.

\bibitem{burdick2007sensing}
J.~W. Burdick, N.~du~Toit, A.~Howard, C.~Looman, J.~Ma, R.~M. Murray, and
  T.~Wongpiromsarn, ``Sensing, navigation and reasoning technologies for the
  darpa urban challenge,'' tech. rep., California Institute of Technology and
  Jet Propulsion Lab, 2007.

\bibitem{pnueli1977temporal}
A.~Pnueli, ``The temporal logic of programs,'' in {\em 18th Annual Symposium on
  Foundations of Computer Science (sfcs 1977)}, pp.~46--57, IEEE, 1977.

\bibitem{baier2008principles}
C.~Baier and J.-P. Katoen, {\em Principles of model checking}.
\newblock MIT press, 2008.

\bibitem{avpsim_video}
J.~Graebener, T.~Phan-Minh, J.~Yan, Q.~Zhao, and R.~M. Murray, ``Automated
  valet parking simulation \url{https://youtu.be/dtDz9zlj46w},'' mar 2020.

\bibitem{trio}
N.~Smith, ``Trio: a friendly python library for async concurrency and i/o,''
  {\em https://trio.readthedocs.io/en/latest/, accessed 03/24/2020}, 2017.

\bibitem{fliess1995flatness}
M.~Fliess, J.~L{\'e}vine, P.~Martin, and P.~Rouchon, ``Flatness and defect of
  non-linear systems: introductory theory and examples,'' {\em International
  Journal of Control}, vol.~61, no.~6, pp.~1327--1361, 1995.

\bibitem{cowlagi2011hierarchical}
R.~V. Cowlagi and P.~Tsiotras, ``Hierarchical motion planning with dynamical
  feasibility guarantees for mobile robotic vehicles,'' {\em IEEE Transactions
  on Robotics}, vol.~28, no.~2, pp.~379--395, 2011.

\bibitem{sakai2018pythonrobotics}
A.~Sakai, D.~Ingram, J.~Dinius, K.~Chawla, A.~Raffin, and A.~Paques,
  ``Pythonrobotics: a python code collection of robotics algorithms,'' {\em
  arXiv preprint arXiv:1808.10703}, 2018.

\bibitem{schurmann2017guaranteeing}
B.~Sch{\"u}rmann and M.~Althoff, ``Guaranteeing constraints of disturbed
  nonlinear systems using set-based optimal control in generator space,'' {\em
  IFAC-PapersOnLine}, vol.~50, no.~1, pp.~11515--11522, 2017.

\bibitem{filippidis2016control}
I.~Filippidis, S.~Dathathri, S.~C. Livingston, N.~Ozay, and R.~M. Murray,
  ``Control design for hybrid systems with tulip: The temporal logic planning
  toolbox,'' in {\em 2016 IEEE Conference on Control Applications (CCA)},
  pp.~1030--1041, IEEE, 2016.

\end{thebibliography}

\end{document}